\documentclass[final,twoside,11pt]{entics} 
\usepackage{enticsmacro}
\usepackage{graphicx}
\usepackage{caption}
\usepackage{subcaption}
\usepackage[all]{xy}

\def\conf{MFPS 2023} 	
\def\myconf{mfps39}
\volume{3}			
\def\volu{3}			
\def\papno{7}			
\def\lastname{Dixon, Murawski} 
\def\runtitle{Saturating automata for game semantics} 
\def\mydoi{12277}
\def\copynames{A. Dixon, A. Murawski}    

\usepackage{xcolor}
\usepackage[all]{xy}
\usepackage{extarrows}
\usepackage{amssymb}
\usepackage{stmaryrd}
\usepackage{pst-node}
\usepackage{mathabx}
\usepackage{bussproofs}
\usepackage{sidenotes}

\usepackage{pgf}
\usepackage{svg}

\usepackage{tikz}
\usetikzlibrary{arrows}
\usetikzlibrary{positioning}

\usepackage[normalem]{ulem}

\newcommand\trace[1]{\mathit{Tr}(#1)}
\newcommand\lang[1]{\mathit{L}(#1)}
\newcommand\abra[1]{\langle #1 \rangle}
\newcommand\cutout[1]{}
\newcommand\pmor{\rightharpoonup}

\newcommand\imax{\mathit{max}}

\newcommand\alp[1]{\mathcal{T}_{#1}}

\newcommand\sata{\mathsf{SATA}}

\newcommand\cn[1]{C^{(#1)}}

\newcommand\ees[1]{e^{(#1)}}

\newcommand\cs[1]{c^{(#1)}}
\newcommand\ds[1]{d^{(#1)}}
\newcommand\es[1]{e^{(#1)}}

\newcommand\CS[1]{C^{(#1)}}
\newcommand\DS[1]{D^{(#1)}}
\newcommand\ES[1]{E^{(#1)}}

\newcommand\N{\mathbb{N}}

\newcommand\sem[1]{{\llbracket}{#1}{\rrbracket}}
\newcommand\seq[2]{{#1} \vdash {#2}}

\newcommand\rarr\rightarrow
\newcommand\dom[1]{\mathsf{dom}(#1)}
\newcommand\makeset[1]{\{#1\}}
\newcommand\word[1]{\textrm{#1}}

\newcommand\trans[1]{{\xlongrightarrow{#1}}}
\newcommand\Aut{\mathcal{A}}

\newcommand\D{\mathcal{D}}
\newcommand\pred[1]{\mathit{pred}(#1)}
\newcommand\predc{\mathit{pred}}
\newcommand\predd[2]{\mathit{pred}^{#1}(#2)}
\newcommand\add{\mathrm{ADD}}
\newcommand{\del}{\mathrm{DEL}}
\newcommand\eps{\mathrm{EPS}}
\newcommand\mop[3]{#1,#2,#3}

\newcommand\fica{\mathsf{FICA}}
\newcommand\iatype[1]{{\bf #1}}
\newcommand\expt{\iatype{exp}}
\newcommand\comt{\iatype{com}}
\newcommand\vart{\iatype{var}}
\newcommand\semt{\iatype{sem}}
\newcommand\ctx{\mathcal{C}}
\newcommand\iaterm[1]{{\bf #1}}

\newcommand\fa{\mathbf{FA}}
\newcommand\oa{\mathbf{OA}}
\newcommand\pq{\mathbf{PQ}}
\newcommand\mul[1]{\mathfrak{M}(#1)}
\newcommand{\aasg}{\,\raisebox{0.065ex}{:}{=}\,}

\newcommand\deref[1]{!#1}
\newcommand\while[2]{\iaterm{while}\,#1\,\iaterm{do}\,#2}
\newcommand\cond[3]{\iaterm{if}\,#1\,\iaterm{then}\,#2\,\iaterm{else}\,#3}
\newcommand\skipcom{\iaterm{skip}}
\newcommand\divcom{\iaterm{div}}
\newcommand\newin[2]{\iaterm{newvar}\,#1\,\iaterm{in}\,#2}
\newcommand\newvar[2]{\iaterm{newvar}\,#1\,\iaterm{in}\,#2}
\newcommand\newsem[2]{\iaterm{newsem}\,#1\,\iaterm{in}\,#2}
\newcommand\grb[1]{\iaterm{grab}(#1)}
\newcommand\rls[1]{\iaterm{release}(#1)}
\newcommand\mkvar[2]{\iaterm{mkvar}(#1,#2)}
\newcommand\mksem[2]{\iaterm{mksem}(#1,#2)}
\newcommand\arop[1]{\mathbf{op}(#1)}

\newcommand\mem{\mathcal{V}}

\newcommand{\astep}[4]{\mem\vdash #2,\,#1 &\longrightarrow & #4,\,#3}
\newcommand{\step}[4]{\mem\vdash #2,\,#1\longrightarrow #4,\,#3}
\newcommand{\sqsubsim}{\,\raisebox{-.5ex}{$\stackrel{\textstyle\sqsubset}{\scriptstyle{\sim}}$}\,}

\newcommand\parc{||}

\newcommand\moveset{{\mathcal{M}}}
\newcommand\move[1]{\mathsf{#1}}
\newcommand\mread{\move{read}}
\newcommand\mwrite[1]{\move{write}(#1)}
\newcommand\mok{\move{ok}}
\newcommand\mrun{\move{run}}
\newcommand\mdone{\move{done}}

\newcommand\mrls{\move{rls}}
\newcommand\mgrb{\move{grb}}
\newcommand\mq{\move{q}}
\newcommand\mi{\move{i}}

\newcommand{\comp}[1]{\textsf{comp}(#1)}

\newcommand\justf[3][]{\nccurve[arrowsize=3pt,nodesep=.5pt,offsetB=-2pt,linewidth=0.8pt,angleA=110,angleB=30,linecolor=darkgray#1]{->}{#2}{#3}}

\newcommand\justh[3][]{\nccurve[arrowsize=3pt,nodesep=.5pt,offsetB=-2pt,linewidth=0.8pt,angleA=130,angleB=30,linecolor=darkgray#1]{->}{#2}{#3}}

\newcommand\justj[3][]{\nccurve[arrowsize=3pt,nodesep=.5pt,offsetB=-2pt,linewidth=0.8pt,angleA=160,angleB=30,linecolor=darkgray#1]{->}{#2}{#3}}

\newcommand\justn[4][]{\nccurve[arrowsize=3pt,nodesep=.5pt,offsetB=-2pt,linewidth=0.8pt,angleA=#4,angleB=30,linecolor=darkgray#1]{->}{#2}{#3}}

\usepackage[capitalise,noabbrev]{cleveref}

\begin{document}
\begin{frontmatter}
  \title{Saturating {A}utomata for {G}ame {S}emantics} 						
  \author{Alex Dixon\thanksref{a}}	
   \author{Andrzej S. Murawski\thanksref{b}\thanksref{c}}		
   \address[a]{Department of Computer Science\\ University of Warwick\\				
    Coventry, UK}  			
    \thanks[c]{This research was funded in whole or in part by EPSRC EP/T006579. For the purpose of Open Access, the author has applied a CC BY public copyright licence to any Author Accepted Manuscript (AAM) version arising from this submission.
   } 				
  \address[b]{Department of Computer Science\\University of Oxford\\
    Oxford, UK} 
        \begin{abstract}

Saturation is a fundamental game-semantic property satisfied by strategies that interpret higher-order concurrent programs.
It states that the strategy must be closed under certain rearrangements of moves, and corresponds to the intuition
that program moves (P-moves) may depend only on moves made by the environment (O-moves).

We propose an automata model over an infinite alphabet, called saturating automata, for which all accepted languages 
are guaranteed to satisfy a closure property mimicking saturation.

We show how to translate the finitary fragment of Idealized Concurrent Algol ($\fica$) into saturating automata,
confirming their suitability for modelling higher-order concurrency.
Moreover, we find that, for terms in normal form, 
the resultant automaton has linearly many 
transitions and states with respect to term size, 
and can be constructed in polynomial time.
This is in contrast to earlier
attempts at finding automata-theoretic models of $\fica$, 
which did not guarantee saturation 
and involved an exponential blow-up during translation,
even for normal forms.
\end{abstract}
\end{frontmatter}


\section{Introduction}

Game semantics is a versatile modelling theory that interprets computation 
as interaction between two players, called O (Opponent) and P (Proponent).
The two players represent the environment and the program respectively,
so programs can be interpreted as strategies for P.
Although initially game models concerned functional sequential computation, 
notably the language PCF~\cite{AJM00,HO00},
it did not take long for the methodology to be extended 
to other programming constructs such as state~\cite{AM97a,AHM98}, control operators~\cite{Lai97},
and, soon afterwards, concurrency.
Some of the game models were presented in the interleaving tradition of models of concurrency~\cite{Lai01,Lai06,GM08}, 
while others were built in the spirit of partial-order methods (true concurrency)~\cite{CCRW17}.

In the interleaving  approach, the aim is to construct strategies in such a way that
they will contain all possible sequential observations of parallel interactions.
Within game semantics, this led to the realisation that strategies must be closed 
under certain rearrangements of moves, to reflect the limited power of programs
to observe and control the actual ordering of concurrent actions.
Critically, a program can wait until an environment action occurs before proceeding,
but it does not have any influence over environment actions or its own concurrent actions beyond those stipulated by the game.
To express this constraint, one requires that strategies
should be closed under certain move swaps. More specifically,
consecutive $m_1 m_2$ can be swapped as long as the swap still leads to a valid play and
it is \emph{not} the case that $m_1$ is an O-move and $m_2$ is a P-move.

In game semantics, this condition first appeared in a model of Idealized CSP~\cite{Lai01},
and was named \emph{saturation} in~\cite{GM08}.
In game models based on event structures~\cite{CCRW17}, 
an analogous condition can be expressed more directly using event structures with polarity.
Variants of saturation also occur in other contexts in the theory of concurrency.
For example, they have been used to describe propagation of signals across wires in
delay-insensitive circuits~\cite{Udd86} or to
specify the relationship between input and output in asynchronous
systems with channels~\cite{JJH90}.  

More recently, there have been attempts at defining automata-theoretic formalisms that provide support
for representing plays in concurrent game semantics~\cite{DLMW21,DLMW21b}. At the technical level, plays are
sequences of moves connected by pointers, which poses a challenge for standard automata theory based on
finite alphabets. However, an infinite alphabet is ideal for this purpose,
especially if it has tree structure, so that the parent relation (link from child to parent) can provide 
a means of representing game-semantic pointers.
Although the proposed formalisms were shown to accommodate the game semantics of higher-order concurrent programs, 
notably, that of a finitary version of Idealized Concurrent Algol ($\fica$)~\cite{GM08}, they do not capture natively
the saturation condition: 
in addition to interpretations of $\fica$ terms (which are guaranteed to satisfy saturation),
they are also capable of accepting many other languages, which need not be closed under any kind of swaps.

In contrast, in this paper, we define an automata model over infinite alphabets,
called \emph{saturating automata}, for which any accepted language
is guaranteed to satisfy (a language variant of) the saturation condition.
It is achieved through carefully tailored transitions, which in particular
restrict the way that siblings may communicate with each other through parents,
and minimise direct communication between other generations.

The new design turns out to bring another technical advantage over existing translations.
Saturating automata corresponding to $\fica$ terms in normal form have linearly many states and transitions
(with respect to term size), and can be generated in at most quadratic time.
This is an improvement over the exponential complexity inherent in earlier translations,
which was due to either the fact that memory was modelled through control states~\cite{DLMW21}
or the use of product constructions to handle parallel composition~\cite{DLMW21b}.
In view of the ubiquity of the saturation condition, we believe that this makes
saturating automata into a point of interest in the design space of automata models,
which deserves further study in connection with game semantics or other areas mentioned above.

\subsection*{Related work} 

In addition to the papers already mentioned,
the combination of game semantics and automata theory over infinite alphabets
appeared in research into sequential computation, e.g.
to handle call-by-value computation with state~\cite{CHMO15,CBMO19}, ground references~\cite{MT18} and objects~\cite{MRT15}.
On the concurrent front, Petri-net-like devices have recently been proposed to interpret higher-order concurrency along with a correspondence to game semantics~\cite{CC23}. 

More broadly, our results are related to encodings of higher-order computation
in process calculi~\cite{San93,RS99,BHY01} (where the role of infinite alphabets would be played by a set of names)
and to abstract machines~\cite{LTY17}.
It would also be interesting to find connections between our work and trace theory over 
partially commutative alphabets~\cite{CF69,Maz78}, though there the commutation relation is typically symmetric, 
unlike in our case.


\section{Finitary Idealised Concurrent Algol ($\fica$)}
\label{sec:fica}

\begin{figure}[t]
\centering
  \AxiomC{$\phantom{\beta}$}
 \UnaryInfC{$\Gamma\vdash\skipcom:\comt $}
\DisplayProof\quad
  \AxiomC{$\phantom{\beta}$}
 \UnaryInfC{$\Gamma\vdash\divcom_\theta:\theta $}
\DisplayProof\quad
  \AxiomC{$0\le i\le\imax$}
 \UnaryInfC{$\Gamma\vdash i:\expt$}
\DisplayProof\quad
  \AxiomC{$\seq{\Gamma}{M:\expt}$}
  \UnaryInfC{$\seq{\Gamma}{\arop{M}:\expt}$}
  \DisplayProof\\[2ex]
  \AxiomC{$\Gamma\vdash M:\comt$}
  \AxiomC{$\Gamma\vdash N:\beta$}
  \BinaryInfC{$\Gamma \vdash M;N:\beta$}
  \DisplayProof\quad
    \AxiomC{$\Gamma\vdash M:\comt$}
  \AxiomC{$\Gamma\vdash N:\comt$}
  \BinaryInfC{$\Gamma \vdash M\parc N:\comt$}
  \DisplayProof\\[2ex]
      \AxiomC{$\Gamma\vdash M:\expt$}
  \AxiomC{$\Gamma\vdash N_1,N_2:\beta$}
  \BinaryInfC{$\Gamma\vdash \cond{M}{N_1}{N_2}:\beta$}
  \DisplayProof\qquad
    \AxiomC{$\Gamma\vdash M:\expt$}
  \AxiomC{$\Gamma\vdash N:\comt$}
  \BinaryInfC{$\Gamma\vdash \while{M}{N}:\comt$}
  \DisplayProof\\[2ex]
  \AxiomC{$\phantom{\beta}$}
  \UnaryInfC{$\Gamma, x:\theta \vdash x: \theta$}
  \DisplayProof\qquad
  \AxiomC{$\Gamma,x:\theta\vdash M:\theta'$}
  \UnaryInfC{$\Gamma\vdash\lambda x. M:\theta\rarr\theta' $}
  \DisplayProof\qquad
  \AxiomC{$\Gamma\vdash M:\theta\rarr\theta'$}
  \AxiomC{$\Gamma\vdash N:\theta$}
  \BinaryInfC{$\Gamma \vdash M N:\theta'$}
  \DisplayProof\\[2ex]
      \AxiomC{$\Gamma\vdash M:\vart$}
  \AxiomC{$\Gamma\vdash N:\expt$}
  \BinaryInfC{$\Gamma \vdash M\aasg N:\comt$}
  \DisplayProof\quad
  \AxiomC{$\Gamma\vdash M:\vart$}
  \UnaryInfC{$\Gamma \vdash {!}M:\expt$}
  \DisplayProof\quad
      \AxiomC{$\Gamma, x:\vart\vdash M:\comt,\expt$}
  \UnaryInfC{$\Gamma\vdash \newin{x}{M:\comt,\expt}$}
  \DisplayProof\\[2ex]
  \AxiomC{$\Gamma\vdash M:\semt$}
  \UnaryInfC{$\Gamma \vdash \rls{M}:\comt$}
  \DisplayProof\quad
  \AxiomC{$\Gamma\vdash M:\semt$}
  \UnaryInfC{$\Gamma \vdash \grb{M}:\comt$}
  \DisplayProof\quad
  \AxiomC{$\Gamma,s:\semt\vdash M:\comt,\expt$}
  \UnaryInfC{$\Gamma\vdash \newsem{s}{M:\comt,\expt}$}
  \DisplayProof

\caption{$\fica$ typing rules\label{fig:icatypes}}
\end{figure}

Idealised Concurrent Algol~\cite{GM08} is a paradigmatic 
call-by-name language 
combining higher-order computation with imperative constructs 
in the style of Reynolds~\cite{Rey78}, 
extended to concurrency with parallel composition ($\parc$) and binary semaphores.
We consider its finitary variant, $\fica$,
defined over a finite datatype $\makeset{0,\ldots,\imax}$ ($\imax\ge 0$), with no recursion, but with iteration.
Its types $\theta$ are generated by the grammar 
\[
\theta::=\beta\mid \theta\rarr\theta\qquad\qquad
  \beta::=\comt\mid\expt\mid\vart\mid\semt
\]
where 
$\comt$ is the type of commands;
$\expt$ that of $\makeset{0,\ldots,\imax}$-valued expressions;
$\vart$ that of assignable variables;
and $\semt$ that of semaphores.
The typing judgments are displayed in Figure~\ref{fig:icatypes}.
Here, $\skipcom$ and $\divcom_\theta$ are constants representing termination and divergence respectively,
$i$ ranges over $\{0,\ldots,\imax\}$,
and $\mathbf{op}$ represents unary arithmetic operations, such as successor or predecessor (since we work over a finite datatype, operations of bigger arity can be defined using conditionals).
Variables and semaphores can be declared locally via $\mathbf{newvar}$ and $\mathbf{newsem}$.
Variables are dereferenced using $!M$, and semaphores are manipulated using two (blocking) primitives,
$\grb{s}$ and $\rls{s}$, which  grab and release the semaphore respectively. 
We assume that variables are initialised to $0$
and semaphores are initially released.

In reduction rules, it will be convenient to use
the syntax
$\newvar{x\aasg i}{M}$ and $\newsem{x\aasg i}{M}$,
which allows us to specify initial values more flexibly,
i.e. $\newvar{x}{M}$ and $\newsem{x}{M}$ should be viewed as
$\newvar{x\aasg 0}{M}$ and $\newsem{x\aasg 0}{M}$ respectively.

{
The operational semantics is defined using a (small-step) transition
relation $\step{s}{M}{s'}{M'}$, where $\mem$ is a set of variable names
denoting active \emph{memory cells} and \emph{semaphore locks}.
$s,s'$ are states, i.e.\ functions $s,s':\mem\rightarrow\makeset{0,\cdots,\imax}$, and $M,M'$ are
terms.  We write $s\otimes (v\mapsto i)$ for the state obtained by augmenting $s$ with $(v\mapsto i)$, assuming $v\not\in \dom{s}$.
The basic reduction rules are given in Figure~\ref{fig:os},
where $c$ stands for any language constant ($i$ or $\skipcom$)
and $\widehat{\mathbf{op}}:\{0,\cdots,\imax\}\rarr\{0,\cdots,\imax\}$
is the function corresponding to $\mathbf{op}$.
In-context reduction is given by the schemata:

\begin{center}
\AxiomC{$\mem,v\vdash  M[v/x],s\otimes(v\mapsto i)\longrightarrow M',s'\otimes(v\mapsto i') $ \quad $M\neq c$}
\UnaryInfC{$\mem\vdash\newin{x\aasg i}{M},s\longrightarrow \newin{x\aasg i'}{M'[x/v]}, s' $}
\DisplayProof\\[2ex]
\AxiomC{$\mem,v\vdash  M[v/x],s\otimes(v\mapsto i)\longrightarrow M',s'\otimes(v\mapsto i') $\quad $M\neq c$}
\UnaryInfC{$\mem\vdash\newsem{x\aasg i}{M},s\longrightarrow \newsem{x\aasg i'}{M'[x/v]}, s' $}
\DisplayProof\\[2ex]
  \AxiomC{$\step{s}{M}{s'}{M'}$}
  \UnaryInfC{$\step{s}{\mathcal E[M]}{s'}{\mathcal E[M']}$}
  \DisplayProof
\end{center}
where reduction contexts $\mathcal E[-]$ are produced by the
grammar:
\[\begin{array}{rcl}
  \mathcal E[-] &::=& [-] \mid \mathcal E;N
  \mid (\mathcal E\,\parc\, N)
  \mid (M\,\parc\, \mathcal E)
  \mid {\mathcal E} N 
  \mid \arop{\mathcal E} 
  \mid \cond{\mathcal E}{N_1}{N_2}\\
  &&\mid {!}\mathcal E
  \mid \mathcal E\aasg m
  \mid M\aasg\mathcal E
  \mid \grb{\mathcal E}
  \mid \rls{\mathcal E}.
\end{array}\]
\begin{figure}[t]
\begin{center}
$\begin{array}{rclcrcl}
  \astep {s}{\skipcom\parc\skipcom}{s}{\skipcom} &\quad &   \astep {s}{\cond{i}{N_1}{N_2}}{s}{N_1},\quad i\neq 0\\
  \astep {s}{\skipcom;c}{s}{c} &&   \astep {s}{\cond{0}{N_1}{N_2}}{s}{N_2}\\
  \astep {s}{\arop{i}}{s}{\widehat{\mathbf{op}}(i)} &&    \astep{s}{(\lambda x.M) N}{s}{M[N/x]}\\
  \astep {s}{\newin{x\aasg i}c}{s}{c} &&    \astep {s\otimes(v\mapsto i)}{{!}v}{s\otimes(v\mapsto i)}{i}\\
  \astep {s}{\newsem{x\aasg i}c}{s}{c}  &&\astep {s\otimes(v\mapsto i)}{v\aasg i'}{s\otimes(v\mapsto i')}{\skipcom}
  \end{array}$
  \end{center}
  \begin{center}
  $\begin{array}{rcl}
  \astep {s\otimes(v\mapsto 0)}{\grb v}{s\otimes(v\mapsto 1)}{\skipcom}\\
  \astep {s\otimes(v\mapsto i)}{\rls v}{s\otimes(v\mapsto 0)}{\skipcom},\quad i\neq 0\\
   \astep {s}{\while{M}{N}}{s}{\cond{M}{(N;\while{M}{N})}{\skipcom}}
\end{array}$
\end{center}
\caption{Reduction rules for $\fica$}\label{fig:os}
\end{figure}
We say that a term $\seq{}{M:\comt}$ \emph{may terminate},  written $M\Downarrow$, if 
$\emptyset \vdash \emptyset,\,M\longrightarrow^\ast \emptyset,\skipcom$.

\cutout{
Idealized Concurrent Algol~\cite{GM08} also features variable and semaphore constructors, 
called $\textbf{mkvar}$ and $\textbf{mksem}$ respectively,
which play a technical role in the full abstraction argument, similarly to~\cite{AM97a}.
We omit them in the main body of the paper, because
they do not present technical challenges, but they are covered here for the sake of completeness.

\paragraph*{Typing rules}
\[
\AxiomC{$\Gamma\vdash M:\expt\rarr\comt$}
  \AxiomC{$\Gamma\vdash N:\expt$}
  \BinaryInfC{$\Gamma\vdash \mkvar{M}{N}:\vart$}
  \DisplayProof
\quad
 \AxiomC{$\Gamma\vdash M:\comt$}
  \AxiomC{$\Gamma\vdash N:\comt$}
  \BinaryInfC{$\Gamma\vdash  \mksem{M}{N}:\semt$}
  \DisplayProof
\]

\paragraph*{Reduction rules}

\begin{align*}
  \step {s&}{(\mkvar{M}{N})\aasg M'}{s}{M M'}\\
  \step {s&}{{!}(\mkvar{M}{N}}{s}{N}\\
  \step {s&}{\grb {\mathbf{mksem}\,M N}}{s}{M}\\
  \step {s&}{\rls {\mathbf{mksem}\,M N}}{s}{N}
\end{align*}

\paragraph*{$\eta$ rules for $\vart,\semt$}
\[\begin{array}{rcl}
M &\longrightarrow & \mkvar{(\lambda x^\expt. M\aasg x)}{!M}\\
M &\longrightarrow & \mksem{\grb{M}}{\rls{M}}
\end{array}\]
\cutout{
Using $\mathbf{mkvar}$ and $\mathbf{mksem}$,
one can define $\divcom_\theta$ as syntactic sugar using $\divcom=\divcom_\comt$ only.
\[
\divcom_\theta=\left\{
\begin{array}{lcl}
\divcom && \theta=\comt\\
\divcom;0 && \theta=\expt\\
\mkvar{\lambda x^\expt.\divcom}{\divcom_\expt} && \theta=\vart\\
\mksem{\divcom}{\divcom} & &\theta=\semt\\
\lambda x^{\theta_1}.\divcom_{\theta_2} & & \theta=\theta_1\rarr\theta_2\\
\end{array}\right.
\]}
}

$\fica$ terms can be compared using a notion of
\emph{contextual (may-)equivalence}, denoted $\Gamma\vdash M_1\cong M_2$.
Two terms of the same type and with the same free variables 
are equivalent if 
they cannot be distinguished with respect to termination by any context:
for all contexts $\ctx$ such that $\seq{}{\ctx[M_1]:\comt}$, we have
$\ctx[M_1]\!\Downarrow$ if and only if $\ctx[M_2]\!\Downarrow$.
Using game semantics, one can reduce $\cong$ to equality of
the associated sets of complete plays (Theorem~\ref{thm:full}).
\begin{example}\label{ex:term}
Consider the term
\[
\seq{f:\comt\rarr\comt, c:\comt}{
\newvar{x}{(f\,( x\,\aasg{1} )\,\,\parc\,\, \cond{\,\deref{x}}{\,c\,}{\,\divcom}_\comt);\, \deref{x}} : \expt}
\]
The free variable $f$ can be viewed as representing an unknown
function, to be bound to concrete code by a context. Since
we work in a call-by-name setting, that function may evaluate
its argument arbitrarily many times, including none.
If  the function does not use its argument, the value of $x$
will always be $0$ (we assume that local variables are initialised to $0$)
and the term will never terminate, because the right term inside $\parc$
will always diverge, preventing the whole term from terminating.
On the other hand, as long as $f$ evaluates its argument at least once and terminates,
and the right-hand side of $\parc$ is scheduled after the assignment $x\aasg 1$ (and code bound to $c$ terminates) then the whole term will terminate too, returning $1$.
\end{example}
In the next section we sketch the game semantics of $\fica$.


\section{Game semantics\label{sec:gs}}

In this section, we briefly present 
the fully abstract game model for $\fica$ from~\cite{GM08}, which we rely on in the paper.
Game semantics for $\fica$ involves
two players, called Opponent (O) and Proponent (P),
and the sequences of moves made by them can be viewed as interactions between 
a program (P) and a surrounding context (O).
The games are defined using an auxiliary concept of an arena.
\begin{definition}
An \emph{arena} $A$ is a tuple $\langle{M_A,\lambda_A,\vdash_A, I_A}\rangle$, where:
\begin{itemize}
\item $M_A$ is a set of \emph{moves};
\item $\lambda_A:M_A\rarr\makeset{O,P}\times\makeset{Q,A}$
is a function determining for each $m\in M_A$ whether
it is an \emph{Opponent} or a \emph{Proponent move}, 
and a \emph{question} or an \emph{answer};
we write $\lambda_A^{OP},\lambda_A^{QA}$ for the composite
of $\lambda_A$ with respectively the first and second projections;
\item $\vdash_A$ is a binary relation on $M_A$, called \emph{enabling},
satisfying: if $m\vdash_A n$ then $\lambda_A^{OP}(m)\neq\lambda_A^{OP}(n)$ and $\lambda_A^{QA}(m)=Q$;
\item $I_A\subseteq M_A$ is a set of \emph{initial moves} such that $\lambda_A(I_A)\subseteq\{(O,Q)\}$ and $\vdash_A \cap (M_A\times I_A) =\emptyset$ (no enablers).
\end{itemize}
\end{definition}
Note that an initial move must be an O-question (OQ).
In arenas used to interpret base types all questions are initial -
the possible P-answers (PA) are listed below ($0\le\mi\le \imax$).
\[\centering\renewcommand\arraystretch{0.9}\begin{array}{c|c|c}
~\word{Arena}~ & ~\word{OQ}~  & ~\word{PA}~ \\
\hline
\sem{\comt} & \mrun & \mdone \\\hline
\sem{\vart} & \mread & i \\
          & \mwrite{i} & \mok 
          \end{array}\qquad\qquad
          \begin{array}{c|c|c}
~\word{Arena}~ & ~\word{OQ}~  &~\word{PA}~\\
\hline
\sem{\expt} & \mq &   i\\
\hline
\sem{\semt} & \mgrb & \mok \\
          & \mrls & \mok
\end{array}
\]

\begin{figure}
{
\[\renewcommand\arraystretch{0.9}
\begin{array}{rclcrcl}
M_{A\times B}      &=& M_A+M_B &\qquad\qquad & M_{A\Rightarrow B}      &=& M_A+M_B\\[1mm]
\lambda_{A\times B}&=& [\lambda_A,\lambda_B] && \lambda_{A\Rightarrow B}&=& [\abra{\lambda_A^{PO},\lambda_A^{QA}},\lambda_B]
\qquad (\lambda_A^{PO}(m)= O \textrm{ iff } \lambda_A^{OP}(m)=P)\\[1mm]
\vdash_{A\times B} &=& \vdash_A+\vdash_B && \vdash_{A\Rightarrow B} &=& \vdash_A+\vdash_B+\makeset{\,(b,a)\mid b\in I_B\textrm{ and }a\in I_A}\\[1mm]
I_{A\times B} &=& I_A+I_B && I_{A\Rightarrow B} &=& I_B
\end{array}\]
}
\caption{Arena constructions ($+$ and $[\cdots]$ stand for the disjoint union of sets and functions respectively; $\abra{\cdots}$ denotes pairing).}\label{fig:gs1}
\end{figure}
More complicated types are interpreted inductively using
the \emph{product} ($A\times B$) 
and \emph{arrow} ($A\Rightarrow B$) constructions, given in Figure~\ref{fig:gs1}. 

\begin{figure}[t]
\begin{subfigure}{0.55\linewidth}
\[
\begin{array}{c}
A = \sem{\comt\rarr\comt}\times\sem{\comt} \Rightarrow\sem{\expt}\\[2ex]
\xymatrix@C=2mm@R=2mm{O &&&&\mq\ar@{-}[d]\ar@{-}[ld]\ar@{-}[lld]\\
P &&\mrun^f\ar@{-}[d]\ar@{-}[ld] &\mrun^c\ar@{-}[d] & i\\
O &\mrun^{f1}\ar@{-}[d]&\mdone^f &\mdone^c\\
P &\mdone^{f1}}
\end{array}
\]%
\caption{The arena $A$ for the term from Example~\ref{ex:term}.}\label{fig:gs2a}
\end{subfigure}%
\hfill%
\begin{subfigure}{0.55\linewidth}
\vspace{12.75mm} 
\begin{tikzpicture}[->,>=stealth',auto, node distance=1.5cm,
  thick,main node/.style={},player/.style={}, bend right=30]
   \node[main node] (s) {$s_1=$};
  \node[main node] (q) [right of=s] {$\mathsf{q}^{\vphantom{f}}$};
  \node[main node] (rf) [right of=q] {$\mrun^f$};
  \node[main node] (df) [right of=rf] {$\mdone^{f}$};

  \node[player] (pq) [below = 0mm of q] {$O$};
  \node[player] (prf) [below = 0mm of rf] {$P$};
  \node[player] (pdf) [below = 0mm of df] {$O$};

  \path (rf) edge[bend right] node [left] {} (q);
  \path (df) edge[bend right] node [left] {} (rf);
\end{tikzpicture}
\vspace{13mm}
\caption{$s_1$, a short justified sequence over $A$.}\label{fig:gs2b}
\end{subfigure}

\begin{subfigure}{\linewidth}
  \centering
\begin{tikzpicture}[->,>=stealth',auto, node distance=1.5cm,
  thick,main node/.style={},player/.style={}, bend right=30]
   \node[main node] (s) {$s_2=$};
  \node[main node] (q) [right of=s] {$\mq^{\vphantom{f}}$};
  \node[main node] (rf) [right of=q] {$\mrun^f$};
  \node[main node] (rf1) [right of=rf] {$\mrun^{f1}$};
  \node[main node] (df1) [right of=rf1] {$\mdone^{f1}$};
  \node[main node] (rc) [right of=df1] {$\mrun^{c}$};
  \node[main node] (dc) [right of=rc] {$\mdone^{c}$};
  \node[main node] (df) [right of=dc] {$\mdone^{f}$};
  \node[main node] (1) [right of=df] {$1$};

  \node[player] (pr) [below = 0mm of q] {$O$};
  \node[player] (prf) [below = 0mm of rf] {$P$};
  \node[player] (prf1) [below = 0mm of rf1] {$O$};
  \node[player] (pdf1) [below = 0mm of df1] {$P$};
  \node[player] (prc) [below = 0mm of rc] {$P$};
  \node[player] (pdc) [below = 0mm of dc] {$O$};
  \node[player] (pdf) [below = 0mm of df] {$O$};
  \node[player] (p1) [below = 0mm of 1] {$P$};

  \path (rf) edge[bend right] node [left] {} (q);
  \path (rf1) edge[bend right] node [left] {} (rf);
  \path (df1) edge[bend right] node [left] {} (rf1);
  \path (rc) edge[bend right] node [left] {} (q);
  \path (dc) edge[bend right] node [left] {} (rc);
  \path (df) edge[bend right] node [left] {} (rf);
  \path (1) edge[bend right] node [left] {} (q);
\end{tikzpicture}
\caption{$s_2$, a longer justified sequence over $A$.}\label{fig:gs2c}
\end{subfigure}
\caption{Arenas and justified sequences}
\end{figure}
We write $\sem{\theta}$ for the arena corresponding to type $\theta$. In Figure~\ref{fig:gs2a}, we give (the enabling relation of)
the arena
$A=(\sem{\comt\rarr\comt}\times\sem{\comt})\Rightarrow \sem{\expt}$,
which needs to be constructed to interpret the term from Example~\ref{ex:term}.
We use superscripts to distinguish copies of the same move
(the use of superscripts is consistent with our future convention,
which will be introduced in Definition~\ref{def:tags}).

Given an arena $A$, we specify next what it means to be a legal play in $A$.
For a start, the moves that players exchange will have to form
a \emph{justified sequence}, which is a finite sequence of moves
of $A$ equipped with pointers. Its first move is always initial and has no
pointer, but each subsequent move $n$ must have a unique pointer to an
earlier occurrence of a move $m$ such that $m\vdash_A n$.  We say that
$n$ is (explicitly) \emph{justified by} $m$ or, when $n$ is an answer, that
$n$ \emph{answers} $m$.
If a question does not have an answer in a justified sequence, we say
that it is \emph{pending} in that sequence.  
In Figures~\ref{fig:gs2b},~\ref{fig:gs2c} we give two justified sequences $s_1$ and $s_2$
over $A$.

Not all justified sequences are valid.  In order to constitute a legal
play, a justified sequence must satisfy a well-formedness condition
that reflects the ``static'' style of concurrency of our programming
language: any started sub-processes must end before the parent process terminates.
\cutout{any process starting sub-processes must wait for the
children to terminate in order to continue. In game terms: if a
question is answered then that question and all questions justified by
must have been answered (exactly once).} This is formalised as follows,
where the letters $q$ and $a$ to refer to question- and answer-moves
respectively, while $m$ denotes arbitrary moves.
\begin{definition}
The set $P_A$ of  \emph{plays over $A$} 
consists of the justified sequences $s$ over $A$ that satisfy
the two conditions below.
\begin{description}
\item[FORK]: In any prefix $s'= \cdots\rnode{A}{q} \cdots\rnode{B}{m}\justf{B}{A}$ of $s$, the question $q$ must be pending when $m$ is played.
\item[WAIT]: In any prefix $s'= \cdots\rnode{A}{q} \cdots\rnode{B}{a}\justf{B}{A}$ of $s$, all questions justified by $q$ must be answered.
\end{description}
\end{definition}
\cutout{For two shuffled sequences $s_1$ and $s_2$, $s_1\amalg s_2$ denotes
the set of all interleavings of $s_1$ and $s_2$.
For two sets of shuffled sequences $S_1$ and $S_2$,
$S_1\amalg S_2=\bigcup_{s_1\in S_1,s_2\in S_2}s_1\amalg s_2$.
Given a set $X$ of shuffled sequences, we define $X^0=X$, $X^{i+1}=X^i \amalg X$. 
Then $X^\circledast$, called \emph{iterated shuffle} of $X$, is defined to
be $\bigcup_{i\in\N}X^i$. }

It is easy to check that the justified sequences $s_1, s_2$ from~\cref{fig:gs2b,fig:gs2c} are plays.
\begin{remark}\label{rem:swaps}
It is worth noting that the notion of play is stable with respect to swaps of adjacent
moves except when the swaps involve occurrences of moves $m_1 m_2$ 
related by the pointer structure:
$\rnode{A}{m_1}\,\rnode{B}{m_2}\justf{B}{A}$
or $m_1, m_2$ are answers to questions $q_1, q_2$ such that
$q_2$ justifies $q_1$.
\end{remark}
A subset $\sigma$ of $P_A$ is \emph{O-complete} if $s\in \sigma$
and $s o\in P_A$ imply $so\in\sigma$, when $o$ is an O-move.
\begin{definition}
  A \emph{strategy} on $A$, written $\sigma:A$, is a
  prefix-closed O-complete subset of $P_A$.
\end{definition}
\cutout{
Recall that O represents the role of the environment/context in game semantics.
Thus, strategies record all potential environment actions.

The game model of $\fica$ consists of \emph{saturated} strategies only: the saturation
condition stipulates that all possible (sequential) observations of
(parallel) interactions must be present in a strategy: actions of the
environment (O) can always be observed earlier if possible, actions of the
program (P) can be observed later. To formalize this, for any arena
$A$, we define a preorder $\preceq$ on $P_A$, as the least transitive
relation $\preceq$ satisfying 
$s\, o\, m\, s'\preceq s\, m\, o\, s'$ and $s\, m\, p\, s'\preceq s\, p\, m\, s'$
for all $s,s'$,
where $o$ and $p$ are an O- and  a P-move respectively (in the above pairs of plays 
moves on the left-hand-side of $\preceq$ are assumed to have the same justifiers as on the right-hand-side). 
\begin{definition}\label{def:sat}
A strategy $\sigma:A$ is \emph{saturated} iff, for all $s,s'\in P_A$,
if $s\in \sigma$ and $s'\preceq s$ then $s'\in\sigma$.
\end{definition}
\begin{remark}\label{rem:causal}
Definition~\ref{def:sat} states that saturated strategies are stable 
under certain rearrangements of moves.
Note that $s_0\,  p\, o\, s_1\not \preceq s_0\, o\, p\, s_1$, while other move-permutations are allowed.
Thus, saturated strategies express causal dependencies of P-moves on O-moves. This partial-order aspect 
is captured explicitly in concurrent games based on event structures~\cite{CCRW17}.
\end{remark}
}

Suppose
$\Gamma=\{x_1:\theta_1,\cdots, x_l:\theta_l\}$ 
and $\seq{\Gamma}{M:\theta}$ is a $\fica$-term.
Let us write $\sem{\seq{\Gamma}{\theta}}$ for the arena $\sem{\theta_1}\times\cdots\times\sem{\theta_l}\Rightarrow\sem{\theta}$.
In~\cite{GM08} it is shown how to assign 
a strategy on $\sem{\seq{\Gamma}{\theta}}$ to any $\fica$-term
$\seq{\Gamma}{M:\theta}$. 
We write $\sem{\seq{\Gamma}{M}}$ to refer to that strategy.
For example, $\sem{\seq{\Gamma}{\divcom}}=\{\epsilon, \mrun\}$
and $\sem{\seq{\Gamma}{\skipcom}} = \{\epsilon,\mrun,\rnode{A}{\mrun}\, \rnode{B}{\mdone}\justf{B}{A}\}$.
The plays $s_1,s_2$ turn out to belong to the strategy that interprets the term from Example~\ref{ex:term}. 
Given a strategy $\sigma$,
we denote by $\comp\sigma$ the set of non-empty \emph{complete} plays of $\sigma$, i.e. those in which all questions have been
answered. For example, $s_1$ (\cref{fig:gs2b}) is not complete, but $s_2$ (\cref{fig:gs2c}) is.

The game-semantic interpretation $\sem{\cdots}$ can be viewed as a faithful record of all possible interactions between the
term and its contexts.
It provides a fully abstract model in the sense that contextual equivalence is characterized by the sets of non-empty complete plays.
\begin{theorem}[\cite{GM08}]\label{thm:full}
\cutout{$\Gamma\vdash M_1\sqsubsim  M_2$ iff
$\comp{\sem{\Gamma\vdash M_1}}\subseteq \comp{\sem{\Gamma\vdash M_2}}$.}
We have $\Gamma\vdash M_1\cong  M_2$ if and only if  $\comp{\sem{\Gamma\vdash M_1}}=\comp{\sem{\Gamma\vdash M_2}}$.
\end{theorem}
The strategies corresponding to $\fica$ terms turn out to be
closed under swaps of adjacent moves as long as
the earlier move is a P-move or the later one is an O-move, and the swap produces a play.
Formally, for any arena $A$, let us define $\succeq\subseteq P_A\times P_A$ 
to be the least preorder satisfying
$s\, m\, o\, s' \succeq s\, o\, m\, s'$ and $s\, p\, m\, s' \succeq s\, m\, p\, s'$,
where $m,o,p$ range over moves, O-moves and P-moves respectively.
In the pairs of plays above, we assume that, during a swap, the justification pointers from the two moves also move with them.

\begin{example}
Consider the following play.
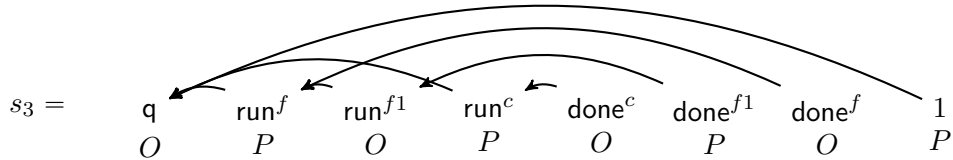
\begin{center}
\begin{tikzpicture}[->,>=stealth',auto, node distance=1.5cm,
  thick,main node/.style={},player/.style={}, bend right=25]
 \node[main node] (s) {$s_3=$};
  \node[main node] (q) [right of=s] {$\mq^{\vphantom{f}}$};
  \node[main node] (rf) [right of=q] {$\mrun^f$};
  \node[main node] (rf1) [right of=rf] {$\mrun^{f1}$};
      \node[main node] (rc) [right of=rf1] {$\mrun^{c}$};
  \node[main node] (dc) [right of=rc] {$\mdone^{c}$};
    \node[main node] (df1) [right of=dc] {$\mdone^{f1}$};
  \node[main node] (df) [right of=df1] {$\mdone^{f}$};
  \node[main node] (1) [right of=df] {$1$};

  \node[player] (pr) [below = -1mm of q] {$O$};
  \node[player] (prf) [below = -1mm of rf] {$P$};
  \node[player] (prf1) [below = -1mm of rf1] {$O$};
    \node[player] (prc) [below = -1mm of rc] {$P$};
  \node[player] (pdf1) [below = -1mm of df1] {$P$};
  \node[player] (pdf) [below = -1mm of df] {$O$};
    \node[player] (pdc) [below = -1mm of dc] {$O$};
  \node[player] (p1) [below = -1mm of 1] {$P$};

  \path (rf) edge[bend right] node [left] {} (q);
  \path (rf1) edge[bend right] node [left] {} (rf);
  \path (df1) edge[bend right] node [left] {} (rf1);
  \path (rc) edge[bend right] node [left] {} (q);
  \path (dc) edge[bend right] node [left] {} (rc);
  \path (df) edge[bend right] node [left] {} (rf);
  \path (1) edge[bend right] node [left] {} (q);
\end{tikzpicture}
\end{center}
\noindent
Observe that $s_2 \succeq s_3$, where $s_2$ is the play from~\cref{fig:gs2c},
because the P-move $\mdone^{f1}$ moved to the right
past a P-move ($\mrun^c$) and an O-move ($\mdone^c$).
In contrast, we do not have $s_3\succeq s_2$,
as this would involve moving a P-move ($\mdone^{f1}$) left
past an O-move ($\mdone^c$).
\end{example}
\begin{example}\label{ex:para}
Consider the plays $s_4, s_5$ given below (in the arena 
$\sem{\comt\rarr\comt\rarr\comt}$), which correspond to
parallel and sequential composition respectively.
Observe that $s_4 \succeq s_5$. Note that the witnessing swap
involves swapping $\mrun^2$ (P-move) with $\mdone^1$ (O-move),
which is permitted by the definition of $\succeq$.
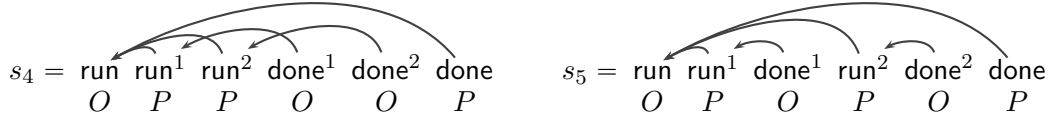
\bigskip
\[\xymatrix@C=0.5mm@R=-.5mm{
s_4= &\rnode{A}{\mrun}&\rnode{B}{\mrun^1}&\rnode{C}{\mrun^2}&
\rnode{D}{\mdone^1} & \rnode{E}{\mdone^2} &\rnode{F}{\mdone}\\
&O&P&P&O&O&P
\justf{B}{A}\justf{C}{A}\justf{D}{B}\justf{E}{C}\justn{F}{A}{130}
}\qquad\xymatrix@C=0.5mm@R=-.5mm{
s_5= &
\rnode{A}{\mrun}&\rnode{B}{\mrun^1}&\rnode{C}{\mdone^1}&
\rnode{D}{\mrun^2} & \rnode{E}{\mdone^2} &\rnode{F}{\mdone}\\
&O&P&O&P&O&P
\justf{B}{A}\justf{C}{B}\justf{D}{A}\justf{E}{D}\justn{F}{A}{130}}
\]
\end{example}
\begin{definition}\label{def:sat}
A strategy $\sigma:A$ is \emph{saturated} if, for all $s,s'\in P_A$,
if $s\in \sigma$ and $s\succeq s'$ then $s'\in\sigma$.
\end{definition}
\begin{remark}\label{rem:causal}
Definition~\ref{def:sat} states that saturated strategies are stable under $\succeq$.
Note that $s\,  o\, p\, s'\not \succeq s\, p\, o\, s'$, while other $o/p$ combinations are allowed in $\succeq$. Thus, saturated strategies allow one to 
express causal dependencies of P-moves on O-moves. 
This aspect of strategies is captured explicitly 
in concurrent games based on event structures~\cite{CCRW17}.
\end{remark}
\begin{theorem}[\cite{GM08}]
For any $\fica$-term $\seq{\Gamma}{M}$, the strategy
$\sem{\seq{\Gamma}{M}}$ is saturated.
\end{theorem}

In the next section we will introduce an automata-theoretic model
for representing plays. 
In contrast to earlier attempts, languages accepted by the automata
will satisfy a language-theoretic equivalent of the saturation condition.


\section{Saturating automata ($\sata$)}\label{sec:sata}

The automata to be introduced will accept 
the so-called data languages, i.e.
languages over an alphabet of the form $\Sigma\times\D$, where
$\Sigma$ is a finite alphabet
and $\D$ is a infinite alphabet of data values.
In our case, the dataset $\D$ will have the structure of a countably infinite forest.
This structure will be helpful when representing game semantics. In particular, it will be used to encode justification pointers and enforce the WAIT condition.
\begin{definition}
$\D$ is a countably infinite set equipped with a function $\predc:\D\rarr\D\cup\{\bot\}$ (the \emph{parent} function) such that the following conditions hold.
\begin{itemize}
\item Infinite branching: $\predc^{-1}(\{d_\bot\})$ is infinite for any $d_\bot\in\D\cup\{\bot\}$.
\item Well-foundedness: for any $d\in\D$, there exists $i\in\N$, called the \emph{level of $d$}, such that $\predc^{i+1}(d)=\bot$.
Level-$0$ data values are called \emph{roots}.
\end{itemize}
We say that $T\subseteq \D$ is a subtree of $\D$ if and only if $T$ is closed ($\forall x \in T \colon \pred{x}\in T\cup\{\bot\}$) and rooted ($\exists!x\in T\colon\pred{x}=\bot$).
\end{definition}
\begin{example}\label{ex:s2word}
Suppose $\Sigma$ consists of moves used in Figure~\ref{fig:gs2a},
$\pred{d_0}=\bot$, $\pred{d_1}=\pred{d_1'}=d_0$ and $\pred{d_2}=d_1$. The play $s_2$ (Figure~\ref{fig:gs2c}) can be represented by the following 
word over $\Sigma\times\D$: 
$(\mq,d_0)
(\mrun^f,d_1)
(\mrun^{f1}, d_2)
(\mdone^{f1},d_2)
(\mrun^c,d_1')
(\mdone^c,d_1')
(\mdone^f,d_1)
(1,d_0)$. Note that the predecessor relation is used to represent justification pointers. Full details of the representation scheme will be given in Section~\ref{sec:fica2sata}.
\end{example}
We use subtrees of $\D$ to represent configurations.
Their nodes will be annotated  with additional information.
We distinguish between odd and even levels to reflect the distinction between the behaviours of the environment (O) and the program (P).
\begin{itemize}
    \item Each even-level node will be annotated
    with a multiset of control states, and zero or more memory cells.
    This information will be allowed to evolve during runs.
    Intuitively, it represents the multiset of states of a group of processes.
    \item Nodes at odd levels will be labelled with single control states, which will not change.
\end{itemize}
In a single transition, the automaton will be able to add or remove leaves 
from its configuration using very limited information. 
When adding a leaf as a child of node $n$, only the state at $n$ will be available. When removing a leaf, in addition to the state at the leaf,
only the parent state will be accessed, if at all.
The automaton will also feature $\epsilon$-transitions, which do not modify the shape of the configuration, 
but can be used to update annotations at even levels, 
while possibly accessing memory cells
at ancestor nodes.

The automata will be parameterized by $k$ and $N$.
The parameter $k$ is the maximal depth of the data used by the automaton, 
while $N$ is the maximal number of memory cells at any node. 
A memory cell will store an element from $V=\{0,\ldots,\imax\}$.
The set of control states will be partitioned into sets
$\CS{i}$, for $0\le i\le k$, dedicated to representing run-time information
at the corresponding level $i$.
\begin{definition}
A \emph{saturating automaton} ($\sata$) is a tuple $\Aut=\abra{\Sigma,k,N,C,\delta}$, where:
\begin{itemize}
\item $\Sigma=\Sigma_{OQ}+\Sigma_{PQ} + \Sigma_{OA} +\Sigma_{PA}$ is a finite alphabet, partitioned into O/P-questions and O/P-answers (we use $q_O,q_P,a_O,a_P$ respectively to range over the elements of the four components);
\item $k\geq 0$ is the depth parameter and $N\ge 0$ is the local memory capacity;
\item $C = \Sigma_{i=0}^k \CS{i}$ is a finite set of \emph{control states}, partitioned into sets $\cn{i}$ of level-$i$ control states;
\item transitions in $\delta$ are partitioned according to their type ($\add$, $\del$ or $\eps$) and level on which they operate;
their shapes are listed below, where $\cs{i},\ds{i},\es{i}\in \CS{i}$
and $\DS{2i},\ES{2i}\in\mul{\CS{2i}}$, where $\mul{X}$ denotes the set of multisets over $X$.
    \begin{itemize}
    \item $\add(2i)$ transitions have the form $\cs{2i-1}\trans{q_O} \DS{2i}$
    or $\dagger\trans{q_O} \DS{0}$ for the special case of $i=0$;
    \item $\add(2i+1)$ transitions have the form $\cs{2i}\trans{q_P} \ds{2i+1}$;
    \item $\del(2i)$ transitions have the form $\DS{2i}\trans{a_P} \dagger$;
    \item $\del(2i+1)$ transitions have the form $\cs{2i+1}\trans{a_O}\ds{2i}$;
    \item $\eps(2i)$ transitions have the form $\DS{2i}\trans{\epsilon}\ES{2i}$;
    \item $\eps(2j,2i)$ transitions read $v \in V$ from memory cell $h \in \{1, \ldots, N\}$ at level $2j \leq 2i$ and update it to $v' \in V$, but do not read the input:
    $(\mop{2j}{h}{v},\cs{2i})\trans{\epsilon}(v',\ds{2i})$. 
    \end{itemize}
\end{itemize}
\end{definition}
\begin{remark}
Observe that O and P can ``act'' only at odd and even levels respectively.
The $\add(2i)$ transitions map exactly onto O-questions from the game semantics. We may view them as spawning a finite number of jobs (hence the use of multisets to represent those jobs' states). Dually, the $\del(2i)$ transition maps onto P-answers which answer those O-questions; correspondingly with WAIT, the $\del(2i)$ transition is only firable when all jobs have reached their ``terminal conditions''. Each job created via $\add(2i)$ can evolve
separately via $\add(2i+1)$ or $\del(2i+1)$, by $\eps(2j,2i)$ (internal state change plus memory operation), or as part of a group via $\eps(2i)$.
\end{remark}
\begin{definition}
A $\sata$ \emph{configuration} is a tuple $(D,E,f,m)$, 
where $D$ is a finite subset of $\D$ (consisting of data values that have been encountered so far),
$E$ is a finite subtree of $\D$ (the shape of the configuration), 
$f:E\rarr \sum\limits_{0<2i-1 \leq k} \CS{2i-1}+ \sum\limits_{0\leq 2i\leq k}\mul{\CS{2i}}$ is  such that
\begin{itemize}
\item if $d$ is a level-$2i$ data value then $f(d)\in \mul{\CS{2i}}$, 
\item if $d$ is a level-$(2i-1)$ data value then $f(d)\in \CS{2i-1}$, 
\end{itemize}
and $m:E\pmor V^N$ is a partial function whose domain is the set of  even-level nodes of $E$.
\end{definition}
A $\sata$ $\Aut$  starts from the empty configuration 
$\kappa_0=(\emptyset,\emptyset,\emptyset,\emptyset)$ 
and proceeds according to its transitions $\delta$, as detailed below.
We write $\kappa=(D,E,f,m)$ and $\kappa'=(D',E',f',m')$
for the current and the successor configurations respectively.

\paragraph*{\rm ADD} We shall have $\kappa\trans{(t,d)} \kappa'$ provided
$t\in\Sigma_{OQ}+\Sigma_{PQ}$, $d\not\in D$, $\pred{d}\in E$,
$D'=D\cup\{d\}$, $E'=E\cup\{d\}$,
and if the transition-specific constraints from the table below
are satisfied\footnote{Given a multiset $(X,\mu:X\rarr\N)$,
we write $x\in_m (X,\mu)$ to mean $\mu(x)>0$. Given
two multisets $(X,\mu_i)$ ($i=1,2$),
we write
$(X,\mu_1)\setminus_m (X,\mu_2)$,
 $(X,\mu_1)\cup_m (X,\mu_2)$ to stand for
$(X,\mu^-)$ and $(X,\mu^+)$ respectively, where
$\mu^-(x)=\max(\mu_1(x)-\mu_2(x),0)$ and
$\mu^+(x)=\mu_1(x)+\mu_2(x)$.
Similarly, $(X,\mu_1)\subseteq (X,\mu_2)$ denotes $\mu_1(x)\le \mu_2(x)$
for all $x\in X$.
}. We write $f[\cdots]$ to extend or update $f$.
\[\begin{array}{c|c|c|c|c}
t&\textrm{transition} & \textrm{pre-condition}  & f' & m'\\
\hline
q_O&\dagger\trans{q_O} \DS{0} & D=\emptyset &\{d\mapsto \DS{0}\} & \{d\mapsto 0^N\}\\
q_O&\cs{2i-1}\trans{q_O}\DS{2i} & f(\pred{d})=\cs{2i-1} &   f [d\mapsto\DS{2i}] & m[d\mapsto 0^N]\\
q_P&\cs{2i}\trans{q_P}\ds{2i+1} & \cs{2i}\in_m f(\pred{d}) &
f \left[\begin{array}{c}\pred{d}\mapsto f(\pred{d})\setminus_m \{\cs{2i}\}\\d\mapsto \ds{2i+1}\end{array}\right] & m\\
\end{array}\]
Note that, in the first two cases, 
memory is initialised at the new node.
In the last case, $\cs{2i}$ is removed from $f(\pred{d})$,
i.e. if a job starts evolving via $\add(2i+1)$, 
it is removed from the list of current jobs.

\paragraph{\rm DEL}
We shall have $\kappa\trans{(t,d)} \kappa'$
provided $t\in\Sigma_{OA}+\Sigma_{PA}$,
$d$ is a leaf in $E$, $D'=D$, $E'=E\setminus\{d\}$, $m'=m$,
and the transition-specific constraints listed below are satisfied.
\[\begin{array}{c|c|c|c}
t&\textrm{transition} & \textrm{pre-condition}  & f' \\
\hline
a_O & \cs{2i+1}\trans{a_O} \ds{2i} & f(d)=\cs{2i+1} &
f [\pred{d}\mapsto f(\pred{d})\cup_m\{\ds{2i}\}]\\
a_P & \DS{2i}\trans{a_P}\dagger & f(d)=\DS{2i} & f\\
\end{array}\]
Note that, in the first case, the leaf will contribute a new state
to the parent node. For simplicity, we do not ``garbage-collect'' $f'$,
since the leaf removal is already recorded via $E'$.

\paragraph{\rm EPS}
We shall have $\kappa\trans{\varepsilon}\kappa'$
provided $D'=D$, $E'=E$ and there exists an even-level datum $d$
satisfying the transition-specific constraints discussed below.
\begin{itemize}
\item For $\DS{2i}\trans{\epsilon}\ES{2i}$,
we require $\DS{2i}\subseteq_m f(d)$,
$f'=f[d\mapsto (f(d)\setminus_m \DS{2i})\cup_m\ES{2i}]$ and $m'=m$.

\item For $(\mop{2j}{h}{v},\cs{2i})\trans{\e}(v',\ds{2i})$,
we require $\cs{2i}\in_m f(d)$ and $m(\mathit{pred}^{2i-2j}(d))(h)=v$,
$f'=f[d\mapsto (f(d)\setminus_m \{\cs{2i}\})\cup_m\{\ds{2i}\}]$
and $m'=m [\mathit{pred}^{2i-2j}(d)(h) \mapsto v' ]$.
\end{itemize}
Note that, in the second case, $m(\mathit{pred}^{2i-2j}(d))(h)$ refers
to the $h$th memory cell of $d$'s ancestor at level $2j$ and only the
content of this cell may be modified by the transition.
\begin{definition}
A \emph{trace} of a $\sata$ $\Aut$ is a word $w \in (\Sigma\times\D)^\ast$ 
such that $\kappa_0\trans{l_1}\kappa_1\dots\kappa_{h-1}\trans{l_{h}}\kappa_{h}$, 
where $\kappa_0=(\emptyset,\emptyset,\emptyset,\emptyset)$,
$l_i\in \{\epsilon\}\cup(\Sigma\times\D)$ ($1\le i\le h$) and
$w=l_1\cdots l_h$. A configuration $\kappa=(D,E,f,m)$  
is \emph{accepting} if $E$ is empty.
A trace $w$ is accepted by $\Aut$ if there is a non-empty sequence of transitions as above with $\kappa_h$  accepting.  
The set of traces (resp. accepted traces) of $\Aut$ is denoted
by $\trace{\Aut}$ (resp. $\lang{\Aut}$).
\end{definition}
It follows that each data value can occur in a trace at most twice. The first occurrence (if any)
must be related to a question, whereas the second one will necessarily be an answer.
The fact that answers can be read only if the corresponding node becomes a leaf 
is analogous to the game-semantic WAIT condition.
Note that $E$ is empty in accepting configurations.
This means that in every word that is accepted,
each question $q_O/q_P$ (corresponding to leaf creation) 
will have a corresponding answer $a_P/a_O$ (corresponding to leaf removal),
and they will be paired up with the same data value. 
Such words resemble complete plays (Theorem~\ref{thm:full})
under the convention that a justification pointer from an answer to a question is represented by using the data value introduced by the question.
Indeed, we will rely on this  when representing plays 
in Section~\ref{sec:fica2sata}.
\begin{example}\label{ex:aut}
{
\newcommand{\ex}[1]{\ensuremath{#1_{\mathsf{ex}}}}
\newcommand{\sset}[1]{\{#1\}}
The $\sata$ $\Aut=\abra{\Sigma, 2, 1, C, \delta}$ specified below recognises complete plays generated by the $\fica$ term from Example~\ref{ex:term} according to the interpretation from \cite{GM08}. 
It is trace- and language-equivalent to the one that would be derived by the translation given in the proof of Theorem~\ref{thm:fica2sata}, though the representation here is made more concise. We use $l$ and $r$ to denote control states corresponding to the left and right subterms of the parallel composition in Example~\ref{ex:term}. The memory value maintained at level 0 corresponds to the value of the variable $x$.

We have
$\Sigma_{OQ} = \{ \mq, \mrun^{f1} \}$, $\Sigma_{PQ}=\{\mrun^f, \mrun^c\}$,
$\Sigma_{OA}=\{ \mdone^f,\mdone^c \}$, $\Sigma_{PA}=\{\mdone^{f1}, 0,\cdots,\imax\}$,
$\CS{0}=\{l_1^{(0)}, l_2^{(0)}, r_1^{(0)}, r_2^{(0)}, r_3^{(0)}, r_4^{(0)} \}$,
$\CS{1} = \{ l_1^{(1)}, r_1^{(1)}\}$, and
$\CS{2} = \{ l_1^{(2)}, l_2^{(2)} \}$.
$\delta$ is given below.

\begin{tabular}{ll}
\textnormal{ ADD(0), DEL(0):}&
$\dagger \trans{\mq} \sset{ l_1^{(0)}, r_1^{(0)} } 
\qquad 
\sset{ l_2^{(0)}, r_4^{(0)} } \trans{1} \dagger
$
\\
\textnormal{ ADD(1), DEL(1):}&
$
l_1^{(0)} \trans{\mrun^{f}} l_1^{(1)}
\qquad
l_1^{(1)} \trans{\mdone^{f}} l_2^{(0)}
\qquad
r_3^{(0)} \trans{\mrun^{c}} r_1^{(1)}
\qquad
r_1^{(1)} \trans{\mdone^{c}} r_4^{(0)}
$
\\
\textnormal{ ADD(2), DEL(2):}&
$
l_1^{(1)} \trans{\mrun^{f1}} \sset{l_1^{(2)}}
\qquad
\sset{l_2^{(2)}} \trans{\mdone^{f1}} \dagger
$
\\
\textnormal{ EPS(0,0):}&
$(0,1,0,r_1^{(0)}) \trans{\varepsilon} (0,r_2^{(0)})\qquad
(0,1,i,r_1^{(0)}) \trans{\varepsilon} (i,r_3^{(0)}) \quad (0 < i\le \imax)$\\
\textnormal{ EPS(0,2):}&
$(0,1,i,l_1^{(2)}) \trans{\varepsilon} (1,l_2^{(2)})\quad (0\le i\le\imax)$
\end{tabular}
}





{
\newcommand{\ex}[1]{\ensuremath{#1_{\mathsf{ex}}}}
\newcommand{\sset}[1]{\{#1\}}

We give a possible transition sequence for $\Aut$.
For the sake of simplicity, data values from $\mathcal{D}$ will be subscripted with a number corresponding to their level, and superscripted with zero or more primes to distinguish within each level. 
Configurations are denoted as a tree of nodes, reflecting the subtree of $\mathcal{D}$ currently maintained in the automaton. 

Nodes at even levels $2i$ are written $d(X)$ or $d(X, v)$, where $d$ is a level-$2i$ data value, $X \in \mul{C^{(2i)}}$ and $v$ represents the memory value maintained at that node
(in this case always a single number).
Nodes at odd levels $2i-1$ have the form $d(X)$, where $d$ is a level-$(2i-1)$ data value and $X \in C^{(2i-1)}$.
The complete transition sequence is given in Figure~\ref{fig:trans-seq-s2}.
It witnesses the acceptance of a data word corresponding to the play $s_2$ from  Figure~\ref{fig:gs2c}.

{
\tikzset{
 auto,
 node distance=1.25cm,
 anchor=base,baseline= (current bounding box.center)
}
\begin{figure}[t]
    \centering
    \[
    \begin{tikzpicture}
    \node (d0) {$\dagger$};
    \end{tikzpicture}
    ~
    \trans{(\mq, d_0)}
    ~
    \begin{tikzpicture}
    \node (d0) {$d_0\,(\sset{ l_1^{(0)}, r_1^{(0)} }, 0)$};
    \end{tikzpicture}
    ~
    \trans{(\mrun^f, d_1)}
    ~
    \begin{tikzpicture}
    \node (d0) {$d_0\,(\sset{ r_1^{(0)} }, 0)$};
    \node (d1) [below of=d0] {$d_1\,({ l_1^{(1)} })$};
    \path (d0) edge (d1);
    \end{tikzpicture}
    ~
    \trans{(\mrun^{f1}, d_2)}
    ~
    \begin{tikzpicture}
    \node (d0) {$d_0\,(\sset{ r_1^{(0)} }, 0)$};
    \node (d1) [below of=d0] {$d_1\,({ l_1^{(1)} })$};
    \node (d2) [below of=d1] {$d_2\,(\sset{ l_1^{(2)} })$};
    \path (d0) edge (d1);
    \path (d1) edge (d2);
    \end{tikzpicture}
    \]
    \[
    \trans{\epsilon}
    ~
    \begin{tikzpicture}
    \node (d0) {$d_0\,(\sset{ r_1^{(0)} }, 1)$};
    \node (d1) [below of=d0] {$d_1\,({ l_1^{(1)} })$};
    \node (d2) [below of=d1] {$d_2\,(\sset{ l_2^{(2)} })$};
    \path (d0) edge (d1);
    \path (d1) edge (d2);
    \end{tikzpicture}
    ~
    \trans{(\mdone^{f1}, d_2)}
    ~
    \begin{tikzpicture}
    \node (d0) {$d_0\,(\sset{ r_1^{(0)} }, 1)$};
    \node (d1) [below of=d0] {$d_1\,({ l_1^{(1)} })$};
    \path (d0) edge (d1);
    \end{tikzpicture}
    ~
    \trans{\epsilon}
    ~
    \begin{tikzpicture}
    \node (d0) {$d_0\,(\sset{ r_3^{(0)} }, 1)$};
    \node (d1) [below of=d0] {$d_1\,({ l_1^{(1)} })$};
    \path (d0) edge (d1);
    \end{tikzpicture}
    ~
    ~
    \trans{(\mrun^{c}, d_1^\prime)}
    ~
    \begin{tikzpicture}
    \node (d0) {$d_0\,(\emptyset, 1)$};
    \node (d1) at (-1,-1.3) {$d_1\,({ l_1^{(1)} })$};
    \node (d11) at (1,-1.3) {$d^\prime_1\,({ r_1^{(1)} })$};
    \path (d0) edge (d1);
    \path (d0) edge (d11);
    \end{tikzpicture}
    \]\[
    \trans{(\mdone^{c}, d_1^\prime)}
    ~
    \begin{tikzpicture}
    \node (d0) {$d_0\,(\sset{ r_4^{(0)} }, 1)$};
    \node (d1) [below of=d0] {$d_1\,({ l_1^{(1)} })$};
    \path (d0) edge (d1);
    \end{tikzpicture}
    ~
    \trans{(\mdone^{f}, d_1)}
    ~
    \begin{tikzpicture}
    \node (d0) {$d_0\,(\sset{ l_2^{(0)}, r_4^{(0)} }, 1)$};
    \end{tikzpicture}
    ~
    \trans{(1, d_0)}
    ~
    \begin{tikzpicture}
    \node (d0) {$\dagger$};
    \end{tikzpicture}
    \]
    \caption{A transition sequence corresponding to $s_2$ (Figure~\ref{fig:gs2c}).}%
    \label{fig:trans-seq-s2}
\end{figure}
}

}

\end{example}


\section{Saturation}\label{sec:sat}

In this section we define a language variant of saturation and show
that languages traced and accepted by $\sata$ satisfy it.
$d_1, d_2\in \D$ will be called \emph{independent} if neither $d_1=\predd{k}{d_2}$ nor $d_2=\predd{k}{d_1}$ for $k\ge 0$, 
i.e. the data lie on different branches. 
Let $\Sigma_O=\Sigma_{OQ}+\Sigma_{OA}$ and $\Sigma_P=\Sigma_{PQ}+\Sigma_{PA}$.
\begin{definition}
We shall say that $L\subseteq (\Sigma\times\D)^\ast$ is \emph{saturated} iff,
for any $w\in L$ and independent $d_1,d_2$,
$w=w_1 (t_1,d_1)(t_2,d_2) w_2\in L$ implies $w_1 (t_2,d_2)(t_1,d_1) w_2\in L$
whenever $t_1\in\Sigma_P$ or $t_2\in \Sigma_O$.
\end{definition}

\begin{remark}
The condition ``$t_1\in\Sigma_P$ or $t_2\in \Sigma_O$'' is the negation
of ``$t_1\in\Sigma_O$ and $t_2\in\Sigma_P$'', i.e. 
the swap is allowed unless the first letter is from $\Sigma_O$ and the second one from $\Sigma_P$.
Note that this is analogous to the game-semantic saturation condition (Definition~\ref{def:sat}).
The definition above uses independent $d_1,d_2$. It would not make sense
to extend it to any dependent cases: one can show that in such cases
the swap will never result in a trace
.
\end{remark}
To show that saturating automata are bound to produce saturated sets of traces/accepted words, we establish a series of lemmas
about commutativity 
between various kinds of transitions.
\begin{lemma}[$\epsilon O\mapsto O\epsilon$]\label{lem:eo}
If $\kappa_1 \trans{\epsilon} \kappa_2 \trans{(t,d)} \kappa_3$ 
and $t\in\Sigma_O$ then 
$\kappa_1\trans{(t,d)}\kappa_2' \trans{\epsilon} \kappa_3$
for some $\kappa_2'$.
\end{lemma}
\begin{proof}
We need to consider all combinations of the transitions listed below.

\begin{tabular}{c|c}
$\epsilon$ & $O$ \\
\hline
$\DS{2i}\trans{\epsilon}\ES{2i}$ \quad or \quad $(\mop{2j}{h}{v},\cs{2i})\trans{\epsilon}(v',\ds{2i})$ 
& $\cs{2i'-1}\trans{q_O} \DS{2i'}$\quad or \quad $\cs{2i'+1}\trans{a_O}\ds{2i'}$
\end{tabular}

\smallskip

\noindent
Observe that the EPS transitions do not modify states at odd levels
or add nodes. Thus, the $\Sigma_O$ transitions could be fired from $\kappa_1$.
Now note that the $\Sigma_O$ transitions cannot prevent the EPS transitions 
from being executed next, because they do not change states at even levels (though
they add new ones).
\end{proof}
\begin{remark}
The converse to Lemma~\ref{lem:eo} is false.
If a $\Sigma_O$ transition is followed by an $\eps$ transition, it may be impossible to swap them, 
because the latter could rely on states introduced by the former.
\end{remark}
\begin{lemma}[$ P \epsilon \mapsto \epsilon P$]\label{lem:pe}
If $\kappa_1 \trans{(t,d)} \kappa_2 \trans{\epsilon} \kappa_3$ 
and $t\in\Sigma_P$ then 
$\kappa_1\trans{\epsilon}\kappa_2' \trans{(t,d)} \kappa_3$
for some $\kappa_2'$.
\end{lemma}
\begin{proof}
We inspect the shape of the relevant rules, which are listed below.

\begin{tabular}{c|c}
$P$ & $\epsilon$ \\
\hline
$\cs{2i}\trans{q_P} \ds{2i+1}$\quad or\quad $\DS{2i}\trans{a_P} \dagger$ & $\DS{2i'}\trans{\epsilon}\ES{2i'}$\quad or\quad $(\mop{2j}{h}{v},\cs{2i'})\trans{\epsilon}(v',\ds{2i'})$
\end{tabular}

\smallskip
\noindent
Observe that the $\epsilon$ transitions do not depend on any information
introduced by transitions on $\Sigma_P$. Hence, they are executable 
from $\kappa_1$. 
Note also that they will not
destroy any information needed to execute the $\Sigma_P$ transitions when fired, 
as there must already have been enough copies of any information to fire the transitions in the original order.
\end{proof}
\begin{remark}
The converse to Lemma~\ref{lem:pe} is false:
an $\epsilon$ transition may well be followed by a transition
on $\Sigma_P$ that relies on the states introduced by the $\epsilon$ transition.
\end{remark}
\begin{remark}\label{rem:permut}
One can use Lemmata~\ref{lem:eo} and~\ref{lem:pe} to replace sequences of
transitions of the form $\kappa\trans{(t_1,d_1)}(\trans{\epsilon})^\ast\trans{(t_2,d_2)}\kappa'$
with sequences of transitions between the same configurations
such that the transitions on $(t_1,d_1)$ and $(t_2,d_2)$
will be adjacent. 
\begin{itemize}
\item If $t_1\in\Sigma_P$ then, using Lemma~\ref{lem:pe} repeatedly,
one can obtain $\kappa_1(\trans{\epsilon})^\ast \trans{(t_1,d_1)}\trans{(t_2,d_2)} \kappa'$.

\item If $t_2\in\Sigma_O$ then, using Lemma~\ref{lem:eo} this time,
one can obtain
$\kappa_1\trans{(t_1,d_1)}\trans{(t_2,d_2)}(\trans{\epsilon})^\ast \kappa'$.
\end{itemize}
Note that these transformations require either
$t_1\in\Sigma_P$ or $t_2\in\Sigma_O$,
so they cannot be carried out if $t_1\in\Sigma_O$ and $t_2\in\Sigma_P$.
\end{remark}
Next we examine permutability of  consecutive transitions
involving independent data values. 
\begin{lemma}\label{lem:swap}
Suppose $d_1,d_2$ are independent and
$\kappa_1 \trans{(t_1,d_1)} \kappa_2\trans{(t_2,d_2)}\kappa_3$, where $t_1\in\Sigma_P$ or $t_2\in \Sigma_O$. 
Then there exists $\kappa_2'$ such that
$\kappa_1 \trans{(t_2,d_2)} \kappa_2'\trans{(t_1,d_1)}\kappa_3$.
\end{lemma}
\begin{proof}
Recall that non-$\epsilon$ transitions rely only on two consecutive levels of the configuration tree.
Consequently, if $d_1, d_2$ are independent and $\pred{d_1}\neq \pred{d_2}$ then the transitions operate on disjoint regions of the configuration and can be swapped.

Now suppose $\pred{d_1}=\pred{d_2}$ and note that, because of independence, we have $d_1\neq d_2$. Consequently, the transitions must operate at the same level
and concern different children of the same node.
\begin{itemize}
\item If the level is even, we need to consider 
the following combinations of transitions:
$\add(2i)\,\add(2i)$, $\del(2i)\,\add(2i)$,
$\del(2i)\,\del(2i)$ (other cases can be ignored
due to the $t_1\in\Sigma_P$ or $t_2\in\Sigma_O$ constraint).
Recalling that $\add(2i)$ and $\del(2i)$ transitions have the form
$\cs{2i-1}\trans{q_O} \DS{2i}$ and $\DS{2i}\trans{a_P} \dagger$ respectively,
we can confirm that the Lemma holds, because 
the state $\cs{2i-1}$ associated with $\pred{d_1}=\pred{d_2}$ is not modified
and there is no scope for interference between the transitions.

\item If the level is odd, we need to consider 
the following combinations of transitions:
$\add(2i+1)\add(2i+1)$, $\add(2i+1)\del(2i+1)$,
$\del(2i+1)\del(2i+1)$ (other cases can be ignored
due to the $t_1\in\Sigma_P$ or $t_2\in\Sigma_O$ constraint).
Recalling that $\add(2i+1)$ and $\del(2i+1)$ transitions have the form
$\cs{2i}\trans{q_P} \ds{2i+1}$ and $\cs{2i+1}\trans{a_O} \ds{2i}$ respectively,
we can confirm that the Lemma holds, because the transitions will not
interfere. In particular, due to $d_1\neq d_2$,
the $\del(2i+1)$ transition in $\add(2i+1)\del(2i+1)$ cannot use
the state introduced by the preceding $\add(2i+1)$ transition.
\end{itemize}
\end{proof}

\begin{remark}
Note that the ``$t_1\in \Sigma_P$ or $t_2\in\Sigma_O$'' condition is necessary:
in the $\del(2i+1)\,\add(2i+1)$ case (i.e. $a_O q_P$), 
it is possible
for the latter transition to use the target state of the former.
\end{remark}

\begin{thm}
For any $\sata$ $\Aut$, the sets $\trace{\Aut},\lang{\Aut}$ 
are saturated.
\end{thm}
\begin{proof}
Consider $t_1,t_2,d_1,d_2$ such that 
$t_1\in \Sigma_P$ or $t_2\in\Sigma_O$,
$d_1, d_2$ are independent and 
$w_1 (t_1,d_1)(t_2,d_2) w_2\in\trace{\Aut}$.
Thus, there exist $\kappa_1,\kappa_2$ such that
$\kappa_1 \trans{(t_1,d_1)}(\trans{\epsilon})^\ast\trans{(t_2,d_2)}\kappa_2$.
By Remark~\ref{rem:permut}, we can rearrange the transitions to get
$\kappa_1 (\trans{\epsilon})^\ast \trans{(t_1,d_1)}\trans{(t_2,d_2)}
(\trans{\epsilon})^\ast\kappa_2$.
By Lemma~\ref{lem:swap}, we then obtain
$\kappa_1 (\trans{\epsilon})^\ast \trans{(t_2,d_2)}\trans{(t_1,d_1)}
(\trans{\epsilon})^\ast\kappa_2$, i.e.
$w_1 (t_2,d_2)(t_1,d_1) w_2 \in\trace{\Aut}$.
Hence, $\trace{\Aut}$ is saturated.
As $\lang{\Aut}$ is a subset of $\trace{\Aut}$ in which
all questions have answers, $\lang{\Aut}$ is also saturated,
because the  swaps do not affect membership in $\lang{\Aut}$.
\end{proof}

\begin{remark}\label{rem:discussion}
Earlier proposals for automata models of $\fica$~\cite{DLMW21,DLMW21b} 
failed to satisfy saturation.
In retrospect, this was because they allowed for too much communication
between control states at various levels.

Leafy automata~\cite{DLMW21} could  access the whole branch 
of the configuration tree at each transition and modify it during transition.
In particular, each move could access and update the state at the root.
This feature could easily be used to define leafy automata that
are very rigid and not closed under any kind of transition swaps. 
Local leafy automata, also introduced in \cite{DLMW21}, restrict access only to the local part of the branch but still allow communication (thus preventing swaps) between nodes sharing a parent or great-grandparent.

Split automata~\cite{DLMW21b} in turn featured restricted access
to control states at various levels, but their transitions still
allowed for state-based communication between siblings,
through transitions $\cs{2i}\trans{q_P} (\ds{2i},\ds{2i+1})$
and $(\cs{2i},\cs{2i+1})\trans{a_O} \ds{2i}$.
The first rule could be used to create two child nodes in
a specific order only, violating Lemma~\ref{lem:swap} for $t_1, t_2\in \Sigma_P$.
The second rule could be used to delete child nodes in a specific
order only, violating the same lemma for $t_1, t_2\in \Sigma_O$.
Finally, the fact that the two rules can communicate through level $2i$
means that we can make the second one conditional on the first one,
meaning that Lemma~\ref{lem:swap} would be violated for $t_1\in \Sigma_P$
and $t_2\in \Sigma_O$.  Consequently, split automata did not offer
native support for saturation, regardless of the polarity of letters.
\end{remark}

\section{From $\fica$ to $\sata$}\label{sec:fica2sata}

In this section we provide an inductive translation from $\fica$ to $\sata$. The main result states
that, for terms in normal form, the construction can be carried out in quadratic time and the automata 
have linearly many states and transitions (with respect to term size).

First, we describe how to encode justification pointers in plays using data and a special
indexing scheme. Recall from Section~\ref{sec:gs} that, to interpret base types, game semantics uses moves from the set
\[\begin{array}{rcl}
\moveset &=& M_{\sem{\comt}}\cup M_{\sem{\expt}}\cup M_{\sem{\vart}} \cup M_{\sem{\semt}}\\
&=&\{\, \mrun,\, \mdone,\, \mq,\, \mread,\, \mgrb,\, \mrls,\, \mok\, \} \cup \{\,i,\, \mwrite{i}{}\,|\, 0\le i \le \max\,\}.
\end{array}\]
The game-semantic interpretation of 
a term-in-context $\seq{\Gamma}{M:\theta}$ is a strategy over the arena $\sem{\seq{\Gamma}{\theta}}$,
which  is obtained through product and arrow constructions,  starting from arenas corresponding to base types.
As both constructions rely on the disjoint sum, the moves from $\sem{\seq{\Gamma}{\theta}}$ are derived
from  the base types present in types inside $\Gamma$ and $\theta$.
To indicate the exact occurrence of a base type from which each move originates, we will  annotate elements of $\moveset$ with
a specially crafted scheme of superscripts.
Suppose  $\Gamma=\{x_1:\theta_1,\cdots, x_l:\theta_l\}$.
The superscripts will have one of the two forms,  where $\vec{i}\in\N^\ast$ and $\rho\in\N$:
\begin{itemize}
\item $(\vec{i},\rho)$ will represent moves from $\theta$;
\item $(x_v\vec{i}, \rho)$ will  represent moves from $\theta_v$ ($1\le v\le l$).
\end{itemize}
The annotated moves will be written as $m^{(\vec{i},\rho)}$ or $m^{(x_v\vec{i},\rho)}$, where $m\in\moveset$.
We will sometimes omit $\rho$ on the understanding that this represents $\rho=0$.
Similarly, when $\vec{i}$ is omitted, the intended value is~$\epsilon$, e.g. $m$ stands for $m^{(\epsilon,0)}$
and $m^x$ for $m^{(x,0)}$. The next definition explains how the $\vec{i}$ superscripts are
linked to moves from $\sem{\theta}$.
Given $X\subseteq \{ m^{(\vec{i},\rho)} \,|\, \vec{i}\in\N^\ast,\,\rho\in\N\}$ and $y\in \N\cup \{x_1,\cdots, x_l\}$, 
we let $yX = \{m^{(y\vec{i},\rho)}\,|\, m^{(\vec{i},\rho)}\in X\}$.
\begin{definition}\label{def:tags}
Given a type $\theta$, the corresponding alphabet $\alp{\theta}$ is defined as follows
\[\begin{array}{rcll}
\alp{\beta} &=& \{\, m^{(\epsilon,\rho)}\,|\, m\in M_{\sem{\beta}},\,\rho\in\N\,\}\qquad \beta=\comt,\expt,\vart,\semt\\
\alp{\theta_l\rarr\ldots\rarr\theta_1\rarr\beta} &=& \bigcup_{u=1}^l (u\alp{\theta_u}) \cup \alp{\beta}
\end{array}\]
For $\Gamma=\{x_1:\theta_1,\cdots, x_l:\theta_l\}$,
the alphabet $\alp{\seq{\Gamma}{\theta}}$ is defined to be 
$\alp{\seq{\Gamma}{\theta}}=\bigcup_{v=1}^l (x_v \alp{\theta_v}) \cup \alp{\theta}$.
\end{definition}
\begin{example}
Given $\Gamma=\{f:\comt\rarr\comt, c:\comt\}$, we have
\[\alp{\seq{\Gamma}{\expt}}=
\{ \mrun^{(f1,\rho)}, \mdone^{(f1,\rho)}, \mrun^{(f,\rho)}, \mdone^{(f,\rho)}, \mrun^{(c,\rho)},
\mdone^{(c,\rho)}, \mq^{(\epsilon,\rho)}, i^{(\epsilon,\rho)}\,|\,
0\le i\le\imax,\,\rho\in \N\, \}.
\]
\end{example}
Note that $\alp{\seq{\Gamma}{\theta}}$ admits a natural partitioning into $X$-questions and
$X$-answers ($X\in\{O,P\}$), depending on whether the underlying move is an $X$-question or an $X$-answer.
To represent the game semantics of terms-in-context $\seq{\Gamma}{M:\theta}$,
we will represent plays as words over $\Sigma\times\D$, where $\Sigma$ is a finite subset of $\alp{\seq{\Gamma}{\theta}}$.
Only a finite subset will be needed, because $\rho$ will be bounded.

Next we explain how $\rho$ and data will be used to represent
justification pointers.
Because no data value can be used twice with a question,
occurrences of questions correspond to unique data values.
A justification pointer from an answer to a question 
can then be represented simply by pairing up the same data value 
with the answer. Pointers from question-moves will be 
represented with the help
of the index $\rho$.
Initial question-moves do not have a pointer and to represent such questions we simply use $\rho=0$.
To represent moves with justification pointers,
we will rely on $\rho$ on the understanding that
$(m^{(y,\rho)},d)$ represents a pointer to the unique 
question-move that introduced $\mathit{pred}^{\rho+1}(d)$.
The reader may wish to check that Example~\ref{ex:s2word} does follow
this convention (therein $m^x$ stands for $m^{(x,0)}$).
Below we give another example involving $\rho>0$, which
may arise in our translation for certain P-moves.
\begin{example}
The play
$\rnode{Z}{\mq}\,\,\,
\rnode{A}{\mrun^f}\justn{A}{Z}{140}\,\,\,
\rnode{B}{\mrun^{f1}}\justn{B}{A}{140} \,\,\,
\rnode{D}{\mrun^{c}}\justn{D}{Z}{140}$
can be represented by
$(\mrun^{(\epsilon,0)},d_0)$ $(\mrun^{(f,0)},d_1)$ $(\mrun^{(f1,0)}, d_2)$ $(\mrun^{(c,2)},d_3)$, given $\pred{d_{i+1}}=d_i$ ($0\le i\le 2$).
\end{example}

Below we state the main result linking $\fica$ with saturating automata.
Question-moves in this translation are handled with ADD transitions:
$\add(2i)$ and $\add(2i+1)$ correspond to O- and P-questions respectively.
Answer-moves are processed with DEL transitions:
$\del(2i)$ for P-answers and $\del(2i+1)$ for O-answers.

\begin{theorem}\label{thm:fica2sata}
For any $\fica$ term $\seq{\Gamma}{M:\theta}$
there exists a $\sata$ $\Aut_M$ 
over a finite subset of $\alp{\seq{\Gamma}{\theta}}$ 
such that the set of plays represented by words from $\trace{\Aut_M}$ 
is $\sem{\seq{\Gamma}{M:\theta}}$, and $\lang{\Aut_M}$ represents $\comp{\sem{\seq{\Gamma}{M:\theta}}}$.
Moreover, when $M$ is in $\beta$-normal $\eta$-long form\footnote{A term is in $\beta$-normal form if none of its subterms is a $\beta$-redex, and it is $\eta$-long if all occurrences of function identifiers inside the term are fully applied. For every term one can obtain a corresponding $\beta$-normal $\eta$-long form by $\beta$-reduction and $\eta$-expansion; these reductions preserve equivalence.}, 
$\Aut_M$ has linearly many states and transitions, 
and can be constructed in quadratic time.
\end{theorem}
\begin{proof}
It follows from analogous results for the simply-typed $\lambda$-calculus that any $\fica$ term can be reduced to an equivalent term in $\beta$-normal $\eta$-long form.  The argument proceeds by induction on the structure of such forms.
When referring to the inductive hypothesis for a subterm $M_i$,
we use the subscript $i$ to refer to the automata components, 
e.g. $\CS{j}_i$,  $\trans{m}_i$ etc.
In contrast, $\CS{j}$, $\trans{m}$ (without subscripts) will refer to the automaton that is being constructed.
Inference lines $\frac{\qquad}{\qquad}$ indicate that the transitions listed under the line should be added
to the new automaton provided the transitions listed above the line are present in the automaton obtained from the inductive hypothesis.

The following three invariants
that strengthen the inductive hypothesis help us establish
correctness and the requisite complexity. They concern labelled
transitions only.
\begin{itemize}
\item $\oa$ (OA determinacy): 
if $\cs{2i+1}\trans{a_O} \ds{2i}_1$ and
$\cs{2i+1}\trans{a_O} \ds{2i}_2$ then $\ds{2i}_1=\ds{2i}_2$.
\item $\pq$ (PQ pre-determinacy):
if $\cs{2i}_1\trans{q_P}\ds{2i+1}$ and $\cs{2i}_2\trans{q_P}\ds{2i+1}$
then $\cs{2i}_1=\cs{2i}_2$.
\item $\fa$ (final readiness): for every $\DS{0}\trans{a_P}\dagger$,
i.e. where $a_P$ is a \emph{final answer},
whenever the automaton reaches a configuration $(D,E,f,m)$ 
with $\DS{0}\subseteq f(r)$, where $r$ is the root, then the transition can be executed, i.e. $r$ has no children and $f(r)=\DS{0}$.
\end{itemize}
Below we discuss a selection of  cases.
In the first three cases, the corresponding automaton merely needs
to respond to the initial question with a suitable answer or not respond at all
(for $\divcom_\theta$).

\paragraph{$\mathbf{M\equiv\skipcom}$:} $k=0$, $N=0$, $\cn{0}=\{0\}$,
$\delta$ consists of $\dagger\trans{\mrun} \{0\}$ and $\{0\}\trans{\mdone} \dagger$.

\paragraph{$\mathbf{M\equiv i}$:} $k=0$, $N=0$, $\cn{0}=\{0\}$, $\delta$ consists of
$\dagger\trans{\mq} \{0\}$ and $\{0\}\trans{i} \dagger$.

\paragraph{$\mathbf{M\equiv \divcom_\theta}$:} $k=0$, $N=0$, $\cn{0}=\{0\}$.
Supposing $\theta\equiv \theta_l\rarr\cdots\rarr\theta_1\rarr\beta$, 
recall that $I_{\sem{\beta}}$ stands for the set of initial questions in $\sem{\beta}$.
$\delta$ is then given by
\[
\frac{x\in I_{\sem{\beta}}}{
\dagger\trans{x} \{0\} 
}\]
$\pq$ and $\oa$ hold vacuously in the above cases, as they do not feature the relevant transitions.
$\fa$ is also clearly satisfied.

\paragraph{$\mathbf{M\equiv \arop{M_1}}$:}
$k=k_1$, $N=N_1$, $\CS{j}=\CS{j}_1$ ($0\le j\le k$).
In this case, we only need to adjust the final answers,
i.e. we take all transitions for $\Aut_{M_1}$ except $\del(0)$,
and modify the $\del(0)$ transitions as follows.
\[
\frac{\DS{0} \trans{i}_1 \dagger}{ \DS{0}\trans{\widehat{\mathbf{op}}(i)} \dagger}
\]
The above relabelling does not concern transitions relevant 
to $\oa$ and $\pq$, so the properties are simply inherited
from $\Aut_{M_1}$.
$\fa$ holds by appeal to IH.

\paragraph{$\mathbf{M \equiv M_1 || M_2}$:}  
In order to match $\sem{\seq{\Gamma}{M_1 || M_2}}$, this construction needs to interleave $\Aut_{M_1}$ and $\Aut_{M_2}$ while gluing the initial and final moves.
Accordingly, we take $k=\max(k_1,k_2)$, $N=N_1+N_2$, 
$\CS{0}= \CS{0}_1 +\CS{0}_2 + \{\circ_1, \circ_2,\bullet_1,\bullet_2\}$,
$\CS{i}= \CS{i}_1 +\CS{i}_2$ 
($0 < i\le k$, assuming $\CS{i}_u=\emptyset$ for $i>k_u$).
All transitions from $\Aut_{M_1}$ and $\Aut_{M_2}$ other than $\add(0)$, $\del(0)$, $\eps(0,2i)$
are simply embedded into the new automaton.
$\add(0)$ and $\del(0)$ need to be synchronised,
as shown below.
\[
\frac{}{\dagger\trans{\mrun}\{\circ_1,\circ_2\}}
\qquad
\frac{\dagger\trans{\mrun}_u \DS{0}_u\quad u\in\{1,2\}}{\{\circ_u\}\trans{\epsilon}\DS{0}_u}
\qquad
\frac{\DS{0}_u\trans{\mdone}_u \dagger\quad u=1,2}{
\DS{0}_u\trans{\epsilon} \{\bullet_u\}}\qquad
\frac{}{\{\bullet_1,\bullet_2\}\trans{\mdone} \dagger}
\]
$N=N_1+N_2$ reflects the need to combine local memories of the two
automata. This need arises only at level $0$, as memory at other
levels will be disjoint. Consequently, we need to adjust memory
indices for $\eps(0,2i)$ transitions from $\Aut_{M_2}$ only:
\[
\frac{(\mop{0}{h}{v}, \cs{2i})\trans{\epsilon}_1 (v',\ds{2i})}{
(\mop{0}{h}{v}, \cs{2i})\trans{\epsilon} (v',\ds{2i})
}\qquad\qquad
\frac{(\mop{0}{h}{v}, \cs{2i})\trans{\epsilon}_2 (v',\ds{2i})}{
(\mop{0}{N_1+h}{v}, \cs{2i})\trans{\epsilon} (v',\ds{2i})
}.
\]
It follows from IH and the construction that $\fa$ will be preserved for $\mdone$.
$\oa$ and $\pq$ are preserved too, because the construction does not
affect the relevant transitions.

\paragraph{$\mathbf{M\equiv M_1;M_2:\comt}$:}

Here we need to let $\Aut_{M_1}$ run to completion
and then direct the computation to $\Aut_{M_2}$.
We take $k=\max(k_1, k_2)$, $N=N_1+N_2$, 
$\CS{0}= \CS{0}_1 + \CS{0}_2 + \{ \circ \}$,
$\CS{i}= \CS{i}_1 + \CS{i}_2$ ($0 < i\le k$).

We modify the $\add(0)$ and $\del(0)$ transitions as follows.
\[
\frac{\dagger\trans{\mrun}_1 \DS{0}}{
\dagger\trans{\mrun}\DS{0}}
\qquad
\frac{
\DS{0}_1 \trans{\mdone}_1 \dagger}{
\DS{0}_1\trans{\epsilon} \{\circ\}}
\qquad 
\frac{\dagger \trans{\mrun}_2 \DS{0}_2}{
\{\circ\} \trans{\epsilon} \DS{0}_2}
\qquad
\frac{\DS{0}_2 \trans{\mdone}_2 \dagger}{
\DS{0}_2 \trans{\mdone} \dagger}
\]
The remaining transitions are simply copies of
other transitions 
from $\Aut_{M_1}$, $\Aut_{M_2}$,
with the proviso that in $\eps(0,2j)$ transitions 
from $\Aut_{M_2}$ we add $N_1$ to the index of the memory cell
that is accessed.

For correctness, we need to appeal to $\fa$ for $M_1$, which
tells us that reaching a configuration in which 
the root is labelled with $\DS{0}_1$ amounts to the termination of $M_1$.
As before, the construction does not modify transitions
relevant to $\oa$, $\pq$, so the properties are simply inherited
from $M_1$ and $M_2$. $\fa$ follows from IH.

\paragraph{$\mathbf{M\equiv f(M_1)}$:}

This case is interesting, because this is where labelled transitions
are created rather than inherited.
According to~\cite{GM08}, the automaton should start with $\mrun\,\mrun^{f}$ and end with $\mdone^f\,\mdone$.
In the meantime, after the first two moves, it should allow for an arbitrary number of $\mrun^{f1}$s, each of which should trigger a separate copy of $\Aut_{M_1}$, which will terminate with $\mdone^{f1}$.
$\mdone^f$ should be read only when all of the copies are finished.

We discuss the simplest instance $f:\comt\rarr\comt$. 
We take $k=2+ k_1$, $N= N_1$,
$\CS{0}= \{0_\mrun,0_\mdone\}$, 
$\CS{1}=\{1_{\mrun}\}$, $\CS{j+2} = \CS{j}_1$ ($0\le j\le k_1$).
First we add transitions corresponding to calling and returning 
from $f$:
\[
\dagger \trans{\mrun^{(\epsilon,0)}} \{0_\mrun\}\qquad
0_{\mrun}\trans{\mrun^{(f,0)}} 1_{\mrun}\qquad
1_{\mrun} \trans{\mdone^{(f,0)}} 0_{\mdone}\qquad
\{0_\mdone\}\trans{\mdone^{(\epsilon,0)}} \dagger.
\]
In state $1_{\mrun}$ we want to allow the environment to spawn an unbounded number of copies of the strategy for  $\seq{\Gamma}{M_1:\comt}$:
\[
\frac{\dagger\trans{\mrun^{(\epsilon,0)}}_1 \DS{0}_1}{ 1_\mrun  \trans{\mrun^{(f 1, 0)}} \DS{0}_1}\qquad
\frac{\DS{0}_1\trans{\mdone^{(\epsilon,0)}}_1 \dagger}{ 
\DS{0}_1\trans{\mdone^{(f 1, 0)}} \dagger}.
\]
Note that the copies will run
two levels lower than in $\Aut_{M_1}$.

The remaining moves
related to $M_1$ originate from $\Gamma$, i.e.
are of the form
$m^{(x_v\vec{i},\rho)}$, where $(x_v : \theta_v)\in\Gamma$.
The associated transitions need to be embedded into 
the new automaton,
but P-question-moves of the form $m^{(x_v,\rho)}$ (corresponding
to initial moves of $\sem{\theta_v}$) 
need to have their pointer adjusted so that they point
at the move tagged with $\mrun^{(\epsilon,0)}$ (leaving $\rho$ unchanged in this case would mean pointing at $\mrun^{(f 1,0)}$).
To achieve this, it suffices to add $2$ to $\rho$ in this case. Otherwise $\rho$ can remain unchanged, because the pointer structure
is preserved. Below we use $\square_L,\square_R$ to refer to arbitrary 
left/right-hand sides of transition rules.
\[
\frac{ \square_L \trans{m^{(x_v,\rho)}}_1 \square_R\qquad 
\textrm{$m\in\Sigma_Q$}}{\square_L \trans{m^{(x_v,\rho+2)}} \square_R}
\qquad
\frac{ \square_L \trans{m^{(x_v\vec{i},\rho)}}_1 \square_R\qquad\textrm{$\vec{i}\neq \epsilon$ or ($\vec{i}=\epsilon$ and $m\in\Sigma_A$)}}{\square_L \trans{m^{(x_v \vec{i},\rho)}}\square_R}
\]
Memory-related transitions are also copied, while adjusting the depth of the level that is being accessed by adding~$2$:
\[
\frac{(\mop{2j}{h}{v},\cs{2i}) \trans{\epsilon}_1 (v',\ds{2i})}{ 
(\mop{2j+2}{h}{v},\cs{2i}) \trans{\epsilon} (v',\ds{2i})}.
\]
The preservation of $\oa$ and $\pq$ follows from the construction and IH,
as the old transitions are simply copied in and 
relabelled injectively. $\fa$ follows from the shape of the new transitions and IH.

\paragraph{$\mathbf{M\equiv\newvar{x}{M_1}}$:} According to~\cite{GM08}, it suffices
to consider plays from $M_1$ in which $\mread^{(x,\rho)}$ and $\mwrite{j}^{(x,\rho)}$ moves are immediately followed by answers, and
the sequences obey the ``good variable'' discipline (a value that is
read corresponds to the most recently written value).
To implement this recipe in an automaton, 
we add an extra cell at level $0$ to store values of $x$
along with explicit initialisation (to facilitate automata re-use in loops).
To this end, we take
$k=k_1$, $N=N_1+1$, 
$\CS{0}=\CS{0}_1+\{\circ,\bullet\}$,
$\CS{i}= \CS{i}_1$ ($0< i\le k$).
All transitions from $\Aut_{M_1}$ can be copied over except 
$\add(0), \del(0)$ and those with superscripts of the
form $(x,\rho)$, i.e. related to $x$.
$\add(0)$ and $\del(0)$ are handled
as specified below.
\[
\frac{m\in I_{\sem{\beta}}}{
\dagger\trans{m} \{\circ\}}
\qquad
\frac{0\le v\le \imax}{
(\mop{0}{N}{v},\circ)\trans{\epsilon} (0,\bullet)}
\qquad
\frac{\dagger\trans{q}_1 \DS{0}_1}{
\{\bullet\}\trans{\epsilon} \DS{0}_1}
\qquad
\frac{\DS{0}_1 \trans{a}_1 \dagger}{\DS{0}_1 \trans{a} \dagger}
\]
Note that in this case $\beta=\comt,\expt$, 
so $I_{\sem{\beta}}=\{\mrun\}$ or $I_{\sem{\beta}}=\{\mq\}$.

For transitions related to $x$ we proceed as follows.
\[
\frac{\cs{2i}\trans{\mwrite{j}^{(x,\rho)}}_1 \ds{2i+1}\trans{\mok^{(x,0)}}_1 \ees{2i}\quad 0\le v\le \imax
}{
(\mop{0}{N}{v}, \cs{2i}) \trans{\epsilon} (j, \ees{2i}) 
}
\qquad
\frac{\cs{2i}\trans{\mread^{(x,\rho)}}_1 \ds{2i+1}\trans{j^{(x,0)}}_1 \ees{2i}
}{
(\mop{0}{N}{j}, \cs{2i}) \trans{\epsilon} (j, \ees{2i})
}
\]
Thanks to $\oa$, the construction will add (at most) 
$\imax+1$ new transitions for each transition 
$\cs{2i}\trans{\mwrite{j}^{(x,\rho)}}_1 \ds{2i+1}$.
Observe that they have the shape
$(\mop{0}{N}{v}, \cs{2i}) \trans{\epsilon} (j, \ees{2i})$ ($0\le v\le \imax$),
and could be represented succinctly by writing
$(\mop{0}{N}{?}, \cs{2i}) \trans{\epsilon} (j, \ees{2i})$,
where $?$ is a wildcard representing an arbitrary value.
So, each $\cs{2i}\trans{\mwrite{j}^{(x,\rho)}}_1 \ds{2i+1}$ gives rise to a single transition with a wildcard.
As the only modifications on $\eps(2j,2i)$ transitions
are of the kind discussed above (adding to
the first two components, but never values), 
this representation
with wildcards can be propagated in further steps.
Similarly, thanks to $\pq$, each transition 
$\ds{2i+1}\trans{j^{(x,0)}}_1 \ees{2i}$ gives rise
to (at most) one new transition 
$(\mop{0}{N}{j}, \cs{2i}) \trans{\epsilon} (j, \ees{2i})$.

Correctness follows from the fact that it suffices to restrict
the work of $M_1$ to traces in which the relevant moves follow each other~\cite{GM08}. Further, by Lemma~\ref{lem:eo}, it suffices to consider scenarios in which the associated transitions follow each other.
$\oa,\pq$ are preserved, because no new relevant transitions are introduced. $\fa$ follows by appealing to IH.

\paragraph{$\mathbf{M\equiv\newsem{s}{M_1}}$:} This case is very similar to the previous one but 
only two values are possible: $0$ (the initial one) or $1$. Transitions corresponding to grabbing change $0$ to $1$, whereas releasing the semaphore does the opposite.
Thanks to $\oa$ and $\pq$, only one transition will be added for each original transition on $\mgrb$ and $\mrls$ respectively.
\[
\frac{\cs{2i}\trans{\mgrb^{(s,\rho)}}_1  \ds{2i+1}\trans{\mok^{(s,0)}}_1 \ees{2i}
}{
(\mop{0}{N}{0}, \cs{2i}) \trans{\epsilon} (1, \ees{2i})
}
\qquad
\frac{\cs{2i}\trans{\mrls^{(s,\rho)}}_1 \ds{2i+1} \trans{\mok^{(s,0)}}_1 \ees{2i}
}{
(\mop{0}{N}{1}, \cs{2i}) \trans{\epsilon} (0, \ees{2i})
}
\]

\cutout{
 \paragraph{$\mathbf{M\equiv\grb{M_1}}:$} Here we want to perform the same transitions as 
the automaton for $M_1$ would when started from $\mgrb$.
At the same time, $\mgrb$ and the corresponding answer $\mok$ have to be relabelled to $\mrun$ and $\mdone$ respectively.
Consequently, it suffices to preserve all transitions from $\Aut_{M_1}$
except $\add(0)$ and $\del(0)$,
which are modified as follows.
\[
\frac{\dagger\trans{\mgrb}_1 \DS{0}}{\dagger \trans{\mrun} \DS{0}}
\qquad
\frac{\DS{0}\trans{\mok}_1 \dagger}{\DS{0} \trans{\mdone} \dagger}
\]
Clearly, $\oa$, $\pq$ and $\fa$ are inherited in this case.
The case $\mathbf{M\equiv\rls{M_1}}$ is the same as the previous one
but $\mrls$ should be used instead of $\mgrb$.}

\subsection*{Complexity analysis} 

The constructions produce an automaton in which 
there are linearly many states, memory cells and transitions,
with respect to term size. For states, it suffices to observe that each construction adds at most a fixed number of new states
to those obtained from IH. The same applies to memory cells.

The case of transitions is harder, as there are several ways
in which transitions are added to the new automaton.
The easiest case is when a transition is simply copied
from an automaton obtained through IH without any changes
to transition labels. Other cases, represented by inference rules,
are based on single premises (old transitions) and generate
new single transitions. As the old ones are not included
in the new automaton, such rules preserve the number of transitions.
$\newvar$ relies on a rule with two premises but,
as discussed, the outcome could still be viewed as 
a single transition with a wildcard.
Finally, when transitions cannot be traced back to old ones,
their number is always bounded by a constant (we regard $\imax$
as a constant too).

Hence, we can conclude that the number of transitions (possibly with
wildcards) will be linear. Because each transition with a wildcard
represents $\imax+1$ transitions without wildcards,
by instantiating them we still obtain a linear number of transitions.
It is also worth noting that each transition involves at most three states:
whenever sets of states are involved in transitions, 
they contain at most two elements.

Finally, we assess the time complexity of the constructions.
A typical case consists of invoking IH and performing 
a bounded number of linear-time operations on the results 
to implement the constructions, such as
retagging to implement the disjoint sum and relabelling.
The combinations of transitions mentioned in $\iaterm{newvar}$
can also be considered in linear time after some preprocessing
that guarantees constant-time access to incoming and outgoing
transition of a given state.
Overall, this could be viewed as a linear number of linear-time operations,
yielding quadratic time complexity. 
Note that the quadratic bound will not extend to the general case, as the conversion to $\beta$-normal $\eta$-long form can induce a significant blowup in the size of the term.
\end{proof}

\section{Conclusion}
We have introduced saturating automata, a new model of computation over infinite alphabets.
Unlike earlier proposals~\cite{DLMW21,DLMW21b},
the automata accept only languages that satisfy a closure property corresponding to saturation, 
a property that naturally emerges in concurrent interactions between programs and their environment.
Consequently, the automata can be claimed to provide a more intrinsic model of such interactions.

We also showed that saturating automata can be used to represent the game semantics of $\fica$,
a paradigmatic language combining higher-order functions, state and concurrency. 
In contrast to previous translations, one does not incur an
exponential penalty for using saturating automata to interpret $\fica$ terms in normal form,
which further confirms their fit with $\fica$. 
Regarding emptiness testing, one can still obtain decidable cases by imposing restrictions analogous to those for leafy~\cite{DLMW21} and split automata~\cite{DLMW21b}.

The opportunity for further exploration of saturating automata remains, with a view to finding verification routines that can
capitalise on saturation.

\begin{ack}
We thank the anonymous reviewers for helpful comments.
\end{ack}

\bibliographystyle{./entics}
\bibliography{my}


\end{document}